\documentclass{article}
\usepackage{iclr2027_conference,times}
\usepackage{pifont}

\makeatletter
\newcommand{\iclrpreprint}{\iclrfinalcopy\g@addto@macro\maketitle{\lhead{Preprint.}}}
\makeatother

\iclrpreprint

\usepackage[hyphens]{url}
\usepackage{natbib}
\makeatletter
\def\NAT@sort{\@ne}
\def\NAT@make@cite@list{%
  \advance\count@\@ne
  \@safe@activestrue
  \edef\@citeb{\expandafter\@firstofone\@citeb\@empty}%
  \@safe@activesfalse
  \@ifundefined{b@\@citeb\@extra@b@citeb}%
   {\def\NAT@num{A}}%
   {\NAT@parse{\@citeb}}%
  \NAT@ifcat@num\NAT@num
   {\NAT@ifcat@num\NAT@year{\@tempcnta\NAT@year\relax}{\@tempcnta\z@}%
    \multiply\@tempcnta\@m
    \advance\@tempcnta\NAT@num\relax
    \@ifnum{\@tempcnta<\@tempcntb}{%
      \let\NAT@@cite@list=\NAT@cite@list
      \let\NAT@cite@list\@empty
      \begingroup\let\@celt=\NAT@celt\NAT@num@list\endgroup
      \protected@edef\NAT@num@list{%
       \expandafter\NAT@num@celt \NAT@num@list \@gobble @%
      }%
    }{%
      \protected@edef\NAT@num@list{\NAT@num@list \@celt{\number\@tempcnta}}%
      \protected@edef\NAT@cite@list{\NAT@cite@list\@citeb,}%
      \@tempcntb\@tempcnta
    }%
   }%
   {\protected@edef\NAT@nonsort@list{\NAT@nonsort@list\@citeb,}}%
}%
\makeatother

\usepackage{caption}
\usepackage{algorithm}
\usepackage[noend]{algpseudocode}

\usepackage{xspace}
\usepackage{newfloat}
\usepackage{listings}
\usepackage{arydshln}
\usepackage{colortbl}
\usepackage{float}
\DeclareCaptionStyle{ruled}{labelfont=normalfont,labelsep=colon,strut=off}
\floatstyle{ruled}
\newfloat{listing}{tb}{lst}{}
\floatname{listing}{Listing}

\usepackage{graphicx}
\usepackage{wrapfig}
\newlength{\figpair}
\newlength{\figsingle}
\usepackage{subcaption}
\usepackage{booktabs}
\usepackage{multirow}
\usepackage{tabularx}
\usepackage[table]{xcolor}
\usepackage{arydshln}

\definecolor{softgreen}{RGB}{215, 235, 215}
\definecolor{lightgray}{gray}{0.92}
\definecolor{cplx}{RGB}{90,120,158}

\usepackage{amssymb}
\usepackage{mathtools}
\usepackage{amsthm}
\usepackage{nicefrac}

\newcommand{\stok}[0]{\,\nicefrac{\textrm{s}}{\textrm{tok}}}

\theoremstyle{plain}
\newtheorem{theorem}{Theorem}
\newtheorem{proposition}{Proposition}
\newtheorem{lemma}{Lemma}

\theoremstyle{definition}

\theoremstyle{remark}

\makeatletter
\newcommand{\appendixthmnumbering}{%
\@for\@tc:={theorem,proposition,lemma,corollary,definition,assumption,remark,observation}\do{%
  \setcounter{\@tc}{0}%
  \expandafter\xdef\csname the\@tc\endcsname{\noexpand\thesection.\noexpand\arabic{\@tc}}}}
\makeatother

\usepackage[textsize=tiny]{todonotes}

\definecolor{nicecitblue}{rgb}{0.21,0.49,0.74}
\definecolor{hl}{RGB}{225,236,250}
\usepackage[pagebackref,breaklinks,colorlinks,allcolors=nicecitblue]{hyperref}
\usepackage[capitalize,noabbrev]{cleveref}

\crefname{theorem}{Theorem}{Theorems}
\crefname{proposition}{Proposition}{Propositions}
\crefname{lemma}{Lemma}{Lemmas}
\crefname{corollary}{Corollary}{Corollaries}
\crefname{definition}{Definition}{Definitions}
\crefname{assumption}{Assumption}{Assumptions}
\crefname{remark}{Remark}{Remarks}
\crefname{observation}{Observation}{Observations}
\crefname{appendix}{Appendix}{Appendices}
\Crefname{appendix}{Appendix}{Appendices}

\newcommand{\yes}{\textcolor{green!55!black}{\ding{51}}}
\newcommand{\no}{\textcolor{red!70!black}{\ding{55}}}
\newcommand{\dn}{\raisebox{0.15ex}{$\scriptstyle(\downarrow)$}}
\newcommand{\up}{\raisebox{0.15ex}{$\scriptstyle(\uparrow)$}}

\usepackage{tikz}
\newcommand{\circled}[2][0.8]{%
    \tikz[baseline=(char.base), scale=#1]{
        \node[
            circle,
            draw,
            line width=0.6pt,
            inner sep=1.5pt,
            font=\footnotesize\bfseries,
            transform shape,
        ] (char) {#2};
    }%
}

\makeatletter
\def\adl@drawiv#1#2#3{%
        \hskip.5\tabcolsep
        \xleaders#3{#2.5\@tempdimb #1{1}#2.5\@tempdimb}%
                #2\z@ plus1fil minus1fil\relax
        \hskip.5\tabcolsep}
\newcommand{\cdashlinelr}[1]{%
  \noalign{\vskip\aboverulesep
           \global\let\@dashdrawstore\adl@draw
           \global\let\adl@draw\adl@drawiv}
  \cdashline{#1}
  \noalign{\global\let\adl@draw\@dashdrawstore
           \vskip\belowrulesep}}
\makeatother

\let\cite\citep

\title{\centering HEAT: Faster Fully Homomorphic Inference via Approximations-Weights Co-Adaptation}

\author{
    \parbox{\linewidth}{%
    \centering
    Alessandro Zirilli$^{\star, 1}$ \quad Davide Marincione$^{\star, 1}$ \quad Evgenios M. Kornaropoulos$^{2}$  \\ \vspace{.25em}Giuseppe Ateniese$^{2}$\quad Emanuele Rodolà$^{1,3}$
    }%
    \\
    \parbox{\linewidth}{\small \centering
        \vspace{1em}
        $^1$Sapienza University of Rome \quad
        $^2$George Mason University \quad
        $^3$Paradigma
    \\
    [1em]%
    \tt\small \{zirilli, marincione\}@di.uniroma1.it 
    }
}

\begin{document}
\maketitle
\begin{abstract}
    Fully homomorphic encryption (\texttt{FHE}) allows a server to run a language model directly on encrypted user prompts, but current approaches remain prohibitively slow. Ciphertexts natively support only addition, multiplication, and rotation, and multiplications may be composed only to a bounded depth before a costly bootstrapping operation is required to continue. Every nonlinearity must therefore be approximated by an iterative method; each iteration increasing the number of multiplications. A higher iteration count buys precision but exhausts the available depth more frequently and thus triggers more bootstraps, which dominate latency. 
    We introduce \textbf{Homomorphic Encryption-Aware Training (\texttt{HEAT})}, a \emph{fine-tuning} method that makes the per-nonlinearity iteration counts \emph{learnable}, enabling them and the model weights to co-adapt during training. \texttt{HEAT} optimizes iterations with respect to the task objective, allowing the model to adapt to approximation errors encountered during inference without architectural changes or retraining from scratch. We further relate iteration count to quantization bit width and bound, at fixed weights, the gap between our objective and quantization-aware training. On encrypted \texttt{GPT-2} decoding, \texttt{HEAT} reduces iterations by $3.1\times$, bootstraps by $1.6\times$, and end-to-end latency by $1.4\times$, while improving decode agreement over the calibrated encrypted baseline. 
    \begin{center}
    \hspace{0.1cm}\raisebox{-0.2\height}
        {\includegraphics[width=1em,height=1em]{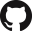}}\hspace{0.1cm} \href{https://github.com/gladia-research-group/heat}{https://github.com/gladia-research-group/heat}
    \end{center}
\end{abstract}

\section{Introduction}\label{sec:intro}
As language models grow in scale, running them locally becomes impractical for most users. Access therefore increasingly takes the form of \textit{inference-as-a-service}, with models hosted on provider-controlled servers and queried remotely. This convenience comes at a privacy cost: prompts and outputs are processed in plaintext and are visible to the provider. Such exposure of users' data is difficult to reconcile with privacy-sensitive domains, where data protection is not only ethically preferable but legally mandated. Fully homomorphic encryption (\texttt{FHE}) protects outsourced language-model inference by enabling encrypted computation on ciphertexts. A client encrypts its prompt, the server evaluates the model without decrypting it, and only the client can recover the output. This is well suited to \textit{inference-as-a-service}, because it preserves a non-interactive workflow (an encrypted request followed by an encrypted response) without requiring mutual trust \cite{Wu2024Ditto:MPC, Li2022MPCFormer:MPC,LuBumbleBee:Transformers, HaoIron:Transformers}.

Running a model under \texttt{FHE}, however, is computationally expensive and suffers from extreme latency. Encrypted computation natively supports only additions, rotations, and multiplications operations, and multiplications can be nested up to a finite maximum \emph{multiplicative depth}, after which expensive \emph{bootstrapping}~\cite{gentry2009fully} is required to restore the multiplicative budget. Moreover, this limited choice of native operations forces nonlinearities to be \emph{approximated by iterative algorithms} that recursively approximate such functions with compositions of addition and multiplication alone, with a precision that increases exponentially with the iteration count. This creates a trade-off between ($i$) approximation precision and ($ii$) inference latency. More iterations produce more accurate estimates, but require more multiplications and thus trigger more bootstrappings, which is the main source of latency during homomorphic inference.
\begin{wrapfigure}{r}{0.50\textwidth}
\centering
    \vspace{-4ex}
    \includegraphics[width=0.50\textwidth]{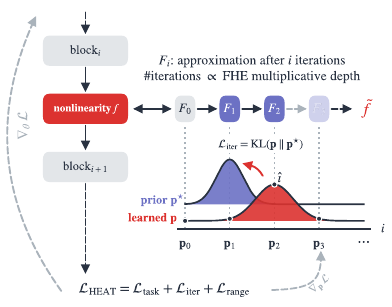}
    \caption{\textbf{\texttt{HEAT} overview.} Each nonlinearity \(f\) is approximated iteratively \(F_0,\ldots,F_n\). During fine-tuning, \texttt{HEAT} learns a distribution \(\mathbf{p}{=}(p_0,\ldots,p_n)\) over iteration counts and evaluates \(\tilde{f}(x){=}\sum_{i=0}^{n} p_i F_i(x)\), regularized toward a prior \(\mathbf{p}^{\star}\) favoring fewer iterations. At deployment, the count is fixed to \(\hat{i}{=}\arg\max_i p_i\), and only \(F_{\hat{i}}(x)\) is evaluated.}
    \label{fig:teaser}
    \vspace{-4ex}
\end{wrapfigure}
For each type of nonlinearity, existing methods \cite{ZhangSecureNon-interactive, JunghoTHOR, yu2026cachemir, Castro2025EncryptedLLM:Encryption, moai} select a single iteration parameterization for the approximation and apply it indiscriminately at every occurrence in the model. Because the chosen approximation must remain accurate for the hardest occurrence, its iteration count is \emph{chosen for the worst case}, which pays for precision it does not need. 
The result is far more bootstraps than the model actually requires. To remedy this, one solution is to tailor the number of iterations to each site. Yet even a per-site adjustment remains suboptimal when each site is held to the same approximation-quality target~\cite{SLOTHE}, since what ultimately matters is the \emph{overall task loss} rather than the per-site accuracy~\cite{LEE2026134016}. Searching the per-layer parameters directly against a task loss deviation proxy~\cite{xie2026atlasautomatedapproximationtransformers} moves the target from the site to the network, but with the weights
frozen, the search can only harvest the tolerance the pretrained model already has.
Alternatively, instead of adjusting approximations, some works redesign or retrain the model around \texttt{FHE} rather than encrypting (and thus approximating) existing architectures unchanged. Some methods train models with polynomial activations~\cite{baruch2023training, ling2025peregrine} or introduce layer architectures that require fewer operations to be approximated~\cite{rho2025encryptionfriendly, ghodsi2020cryptonas, jha2021deepreduce}, while others distill pretrained models into encryption-friendly variants~\cite{park2025powerformer}. However, these approaches alter the model's structure or trigger expensive retraining.
\paragraph{Contribution.} In this work, we introduce {Homomorphic Encryption-Aware Training} (\texttt{HEAT}), a pre-deployment technique that combines per-site iteration selection with fine-tuning (\cref{fig:teaser}). Building on~\citet{banino2021pondernet}, \texttt{HEAT} learns the iteration count at each approximation site jointly with the model's weights. A depth regularization term pushes every site towards fewer iterations, while the task loss retains precision where the model needs it. Because the weights train under the same approximations, they absorb part of the error that shallower solvers introduce, so that each site can be cut further than a fixed model would tolerate. The resulting model is architecturally identical to its plaintext parent but demands substantially less multiplicative depth under \texttt{FHE}.

Nothing guarantees a priori that a network can co-adapt to the error of a shallow iterative solver in place of the exact nonlinearity. We motivate it by relating the error \texttt{HEAT} trains with to the error quantization-aware training (\texttt{QAT}) learns to tolerate. We show that the worst-case error of an $n$-step approximation matches that of a quantizer with bit width $\Theta(2^n)$, \emph{i.e.}\ each iteration roughly doubles the number of correct bits (\Cref{prop:iters-are-bits}). Moreover, for any fixed set of weights, running the network with $n$-step approximants and running it with weights quantized to the matching precision yield losses that differ by a term vanishing doubly exponentially in $n$ (\Cref{thm:heat-is-qat}). At fixed weights, the objective \texttt{HEAT} optimizes, in which every site learns its own count, can thus be read as a per-layer quantization objective, and the learned counts as an allocation of numerical precision across the network layers. Our contributions are threefold:
\begin{itemize}
    \item We introduce \texttt{HEAT}, a fine-tuning method that co-adapts model weights and iteration counts using standard gradient-based optimization. On encrypted \texttt{GPT-2} it reduces the total approximations' iterations by $3.1 \times$.
    
    \item We establish a formal relation between iterative approximation and quantization, and bound, at fixed weights, the difference between the objective under iterative approximations and the \texttt{QAT} objective at the corresponding precision.
    
    \item We evaluate \texttt{HEAT} on teacher forced encrypted decoding against the same model under calibrated approximations. \texttt{HEAT} reduces end-to-end latency from ${\approx}54\stok$ to ${\approx}38\stok$ ($1.4\times$) while raising top-$1$ fidelity to the plaintext model from $70\%$ to $83\%$. 
\end{itemize}
The theoretical results are machine-verified in Lean 4. Code and formal proofs published at \url{https://github.com/gladia-research-group/heat}, model weights available on HuggingFace.
\section{Background}\label{sec:background}
% We introduce the cryptographic and numerical concepts needed for \texttt{HEAT} and its deployment.

\paragraph{Threat Model.}
We have two parties: the client $C$, who wants to prompt the model, and the server $S$, which owns the model. $C$ does not want to reveal the prompt to $S$, and $S$ does not want the model to leave its trust zone\footnote{The weights may be proprietary, in which case $S$ has a commercial interest in keeping them opaque, or open, in which case $S$ may still require that they remain within its infrastructure.}. We operate in the honest-but-curious model for $S$, \emph{i.e.}\, the adversary $S$ follows the protocol but tries to infer information about the privacy-sensitive prompt from the interaction and the computation. We assume the client generates a public--private cryptographic key pair and shares the public key with $S$. The prompt is tokenized and encrypted locally by $C$ before being sent to $S$, so $S$ never sees the plaintext tokens, intermediate activations, outputs, or \textrm{KV} cache, and returns encrypted logits. Additionally, the computation performed by $S$ must be oblivious, \emph{i.e.}\, neither control flow nor memory access patterns can depend on the $C$'s plaintext prompt. Hardware side channels and metadata leakage, such as sequence length, are outside our threat model. %; the latter can be mitigated by padding requests to a fixed length, which we do not assume.

\textbf{Fully Homomorphic Encryption (\texttt{FHE})} allows a server to evaluate arithmetic operations directly on encrypted values. We use \texttt{CKKS}~\cite{CheonHomomorphicNumbers, CheonAEncryption}, an \texttt{FHE} scheme designed for approximate arithmetic over real and complex data. If $c_1 \leftarrow \operatorname{Enc}_{pk}(p_1)$ and $c_2 \leftarrow \operatorname{Enc}_{pk}(p_2)$, then
\begin{equation}\label{eq:encryptedcomputation}
    \operatorname{Dec}_{sk}(c_1 \oplus c_2) \approx p_1+p_2,
    \quad
    \operatorname{Dec}_{sk}(c_1 \otimes c_2) \approx p_1p_2,
    \quad 
    \operatorname{Dec}_{sk}(\operatorname{Rot}_k(c_1)) \approx \pi^k(\mathbf{p}_1),
\end{equation}
where $\oplus$, $\otimes$, and $\operatorname{Rot}_k$ denote encrypted addition, element-wise multiplication, and cyclic slot rotation by $k$ positions, respectively. Multiplications can be composed a finite number of times fixed by the scheme's cryptographic parameters, inducing a quantity known as \emph{multiplicative depth}. Once the available multiplicative depth is exhausted, a \emph{bootstrapping} refreshes the ciphertext and restores its computational capacity~\cite{gentry2009fully}. Bootstrapping costs two to three orders of magnitude more than any other primitive, so multiplicative depth is a primary factor of encrypted-inference latency.

\paragraph{Iterative approximations.} Under \texttt{CKKS}, only additions, multiplications, and rotations can be evaluated directly, so nonlinear operations (division, inverse square root, activation functions\dots) must be approximated with methods exploiting their composition. Common choices include Goldschmidt \cite{goldschmidt_1964} iterations for division, Newton-Raphson iterations for $\nicefrac{1}{\sqrt{x}}$~\cite{Qu2022ImprovementsRoot}, and polynomial approximations based on Chebyshev or Remez fitting. Given a nonlinear function $f(x)$ evaluated at $x$, an iterative method produces increasingly accurate estimates from a seed value $F_0(x)$ that converge to $f(x)$ in the limit
\begin{equation}\label{eq:iterator-limit}
   F_0(x),F_1(x),\ldots,F_n(x)\rightarrow f(x), 
\end{equation}
where $F_i(x)$ is the approximation after $i$ refinement steps and $n$ is the total iteration count. For quadratically convergent methods such as Goldschmidt and Newton-Raphson, the error satisfies
\begin{equation}\label{eq:iterative-convergence}
    \varepsilon_n
    :=
    \sup_{x \in \mathcal{D}} |F_n(x)-f(x)|
    \leq M r^{2^n},
    \qquad M>0,\ r\in(0,1).
\end{equation}
Where $\mathcal{D}$ denotes the set of inputs for which the chosen initialization $F_0(x)$ places the iteration strictly in its convergence interval. The constants $M$ and $r$ depend on the iterative method, its initialization, and the considered domain. Accuracy improves doubly exponentially with $n$, but each additional iteration requires an extra fixed number of multiplications. Therefore, the approximation error decreases rapidly, while its multiplicative depth grows linearly. Practical \texttt{FHE} implementations combine several such methods to balance precision, depth, and 
% ciphertext 
noise~\cite{Qu2022ImprovementsRoot}.

\textbf{Quantization-aware training (\texttt{QAT})} provides the conceptual reference for \texttt{HEAT}. Both expose a model during training to the numerical error it will encounter at deployment, allowing its weights to adapt. Quantization represents weights and, in some settings, activations at lower precision to reduce inference cost. Applying it only after training can degrade accuracy, so \texttt{QAT} inserts low-precision proxies into the forward pass during training or fine-tuning~\cite{jacob2018quantization}. The underlying floating-point parameters are still updated by backpropagation, commonly using a straight-through estimator for non-differentiable rounding. The model therefore learns to tolerate the bounded error introduced by its deployment-time computation.

\section{Method}\label{sec:method}
\texttt{FHE} natively supports only addition, rotation, and multiplication, and multiplications may be composed only a bounded number of times before a costly bootstrapping operation is needed to restore the multiplicative budget. Every other function a model relies on (LayerNorm, Softmax, \dots) must be assembled from these primitives (\Cref{app:approx-operators}). Specifically, during encrypted inference, these functions are replaced by iterative methods that approximate them directly to their components. The parameterization of the iterative methods, and especially the total iterations performed, introduces a trade-off between accuracy and latency. Increasing iterations tends to preserve the model's behavior, producing more accurate approximations, but at the same time means more multiplications and thus more costly bootstrappings. 

Instead of selecting a single iteration parameterization for each approximation or tuning it to a fixed precision target, we propose to model each iteration count as a \emph{parameter} in its own right. We introduce iterative approximation in the plaintext model, during a pre-deployment stage we call \texttt{HEAT}. We fine-tune the iteration counts \emph{jointly with the weights} to meet the final objective, so that the network adapts to the error it will encounter at inference, rather than merely tolerating it. In this manner, each site can be cut further still, since the weights will absorb the error that the iteration reduction introduces.

\subsection{Making iteration selection differentiable}
Let $f$ be a nonlinear function and $F_0,F_1,\ldots,F_n$ the successive results of an iterative scheme converging (in the limit) to $f$, so that $F_i$ is the approximation obtained after $i$ iterations. Substituting $F_i$ for $f$ in the network turns the iteration count $i$ into a \emph{discrete} hyper-parameter that is non-optimizable directly through gradient methods. We make the parameter $i$ differentiable by following \citet{banino2021pondernet}'s adaptive-computation methods. 
Specifically, once the approximations set for future \texttt{FHE} deployment is defined, we implement it in plaintext and, at each approximation site, we introduce a learnable distribution $\mathbf{p}\coloneq (p_0,\ldots,p_n)$ over the iterations, and forward into the next layer the expected approximation under this distribution
\begin{equation}\label{eq:expected-approximation}
\widetilde{f}(x) = \mathbb{E}_{i\sim p}\left[F_i(x)\right] = \sum_{i=0}^{n} p_iF_i(x)\,.
\end{equation}
The loss thereby becomes differentiable in the distribution and in the model weights alike, making the two \emph{co-adapt} during fine-tuning.
Because the $F_i$ are successive states of the same solver, they are all produced in a single execution; during the forward pass we checkpoint only the solvers' inputs and not their intermediate values, which we recompute during the backward, keeping the additional memory costs small (\Cref{app:mem-eff}).
\paragraph{Iteration count penalty.}
To incentivize shallow approximations, for every site $s\in S$, we penalize the Kullback--Leibler (\textrm{KL}) divergence of the learnable distributions $\mathbf{p^{s}}$ from a target negative-binomial prior $\mathbf{p^{s, \star}}$ (whose mode sits at a low $i$) and add it to the task loss
\begin{equation}\label{eq:iteration-loss}
    \mathcal{L}_{\textrm{iter}}
    =
    \frac{1}{|S|}
    \sum_{s\in S}
    \textrm{KL}\!\left(
        \mathbf{p^s}
        \,\middle\|\,
        \mathbf{p^{s,\star}}
    \right),\quad \mathcal{L}_{\texttt{HEAT}} = \mathcal{L}_{\textrm{task}} + \lambda_{\textrm{iter}}\mathcal{L}_{\textrm{iter}} + \lambda_{\textrm{range}}\mathcal{L}_{\textrm{range}}\,.
\end{equation}
The
% resulting 
$\mathcal{L}_\texttt{HEAT}$ introduces a tension between task accuracy and reduction in iterations which, through 
% classical 
gradient descent optimization, finds an equilibrium induced by the coefficients $\lambda_\textrm{range}$ and $\lambda_\textrm{iter}$, where $\mathcal{L}_\textrm{range}$ is the activations' squeeze term, from \citet{zimerman2023converting} and \citet{baruch2023training}.

\paragraph{Adaptive depth target.}\label{sec:adapt}
Applying strong depth pressure from the beginning of training can force the model onto inaccurate approximations before its weights have co-adapted, or driving activations outside iterative methods' convergence domain which results in unstable and slow training. Therefore we initialize each learnable distribution (\Cref{fig:learned-distribution}) such that the expected approximation is heavily biased by high iteration counts' contributions, de facto simulating high precision; but leaving a non-zero tail to let gradient flow through lower iterations' contributions. Likewise, at the onset, the priors' modes are not fixed to encourage extremely coarse approximations, but instead they follow an adaptive back-off that decreases the targets' modes every time the learnable distributions reach it.

\paragraph{Fine-tuning schedule.}\label{sec:heat-schedule}
\texttt{HEAT} is applied to a plaintext pretrained model in three phases. Phase one acts as a warm-up during which we train under $\mathcal{L}_\textrm{task}+\lambda_\textrm{range}\mathcal{L}_\textrm{range}$ alone with no iteration term, reducing activation ranges to prune outlier values, and thus producing a so-called \emph{squeezed} baseline (see~\cref{tab:ablation}). Phase two is the main \texttt{HEAT} functionality: we attach a learnable distribution to every iterative solver's site (\Cref{app:approx-operators}) and optimize with full $\mathcal{L}_{\texttt{HEAT}}$, with the prior's mode adaptively backing off as the learned mode settles on it. Phase three~(\emph{cool-down}) consolidates the result, each approximator is frozen to the mode of its learned distribution $i^{s,\star}{=}\arg\max_i p_i^s$, and the weights are fine-tuned on these \emph{fixed} counts with $\mathcal{L}_\textrm{iter}$ off, consolidating them at the counts to be deployed. At inference, this count is shared by all inputs, as data-dependent execution would leak information through latency, so the encrypted circuit evaluates only $F_{i^{s,\star}}^s(x)$.
\begin{wrapfigure}{r}{0.50\textwidth}
    \centering
    \includegraphics[width=0.50\textwidth]{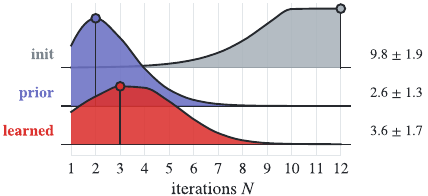}
    \caption{\textbf{Learned distribution.} Evolution of the learnable distribution $\mathbf{p}$ for one \texttt{LayerNorm} approximation. Training shifts probability mass from the full-depth seed toward shallower approximations; deployment selects the learned mode.}
    \label{fig:learned-distribution}
    \vspace{-6ex}
\end{wrapfigure}

\subsection{Approximation Depth as Numerical Precision}
\label{sec:approx-quant}
Neural networks are flexible function approximators that can tolerate a certain degree of noise introduced in the forward stream. A field that has proven this statement widely is quantization, where quantization-aware training (\texttt{QAT}) adapts weights jointly with the task loss and the quantization error. We show that an $n$-step iterative approximation has the same worst-case error as a uniform quantizer whose bit width is set by the iteration count, and that, at fixed weights, a network run with such approximations and the same network run with matched-precision quantized weights have task losses that differ by a term vanishing doubly exponentially in the count. 
More broadly, this analogy lets us study fine-tuning with iterative approximation through the lens of quantization, so that questions about approximation depth become questions about numerical precision.

Consider an $n$-step approximation $F_n$ of a nonlinearity $f$, satisfying \Cref{eq:iterative-convergence}. A uniform $b$-bit quantizer over a range of width $R$ maps elements into representatives taken uniformly with step size $s=R/(2^b-1)$ over $\mathbb{R}$. The representation, for elements in the interval, has a worst-case error $s/2$, and matching this error to $\varepsilon_n$ converts iteration count into an effective bit width. 
%(\Cref{fig:iters-are-bits}).
\begin{proposition}[Iterations are bits]
\label{prop:iters-are-bits}
To obtain the same worst case error from an iterative approximation with a uniform $b$-bit quantizer we need to set:
\begin{equation}
    b
    =
    \log_2\!\left(
        1+\frac{R}{2Mr^{2^n}}
    \right)
    \approx
    2^n\log_2\frac{1}{r}
    +
    \log_2\frac{R}{2M}.
\end{equation}
Hence $b=\Theta(2^n)$, \emph{i.e.}\ each additional iteration approximately doubles the number of correct bits.
\end{proposition}

% \begin{proof}
% Equate the two worst-case errors and solve for $b$. Details in \Cref{app:resolution}.
% \end{proof}

\emph{Proof} and details in \Cref{app:resolution}. It is worth noting that the two mechanisms inject the error from two different sources. Iterative approximation acts directly on the activations, whereas weight quantization perturbs them indirectly through the linear transformation. However, under bounded activations and Lipschitz nonlinearities, the induced per-layer errors can be matched~(\Cref{app:proofs}). Given that models' blocks are $\rho$-Lipschitz functions, and that per-layer error induced by either quantization or approximations is at most $\varepsilon_n$, by induction, the final activation is perturbed by at most $\varepsilon_n\sum_{j=0}^{L-1}\rho^j$. We can thus bound, at fixed weights, the task-loss difference between the two models:
% forward passes:

\begin{theorem}[Matched-precision objective bound]
\label{thm:heat-is-qat}
Let $\mathcal{L}_{\mathrm{approx}}^{(n)}$ denote the task loss of a network using $n$-step iterative approximations, and let
$\mathcal{L}_{\mathrm{QAT}}^{(b)}$ denote the task loss of the same network, with the same weights, under $b$-bit weight quantization. If the loss is
$\xi_{\mathrm{lip}}$-Lipschitz in the final activations and
$b=b^\star(n)=\Theta(2^n)$  is chosen to match the per-layer error, then
\begin{equation}
    \left|
        \mathcal{L}_{\mathrm{approx}}^{(n)}
        -
        \mathcal{L}_{\mathrm{QAT}}^{(b)}
    \right|
    \leq
    2\xi_{\mathrm{lip}}Mr^{2^n}
    \sum_{j=0}^{L-1}\rho^j
\end{equation}
and the bound vanishes doubly exponentially with $n$.
\end{theorem}

\emph{Proof} in \Cref{app:heat-is-qat}. The theorem is stated using a single iteration count (and thus bit-width) across the model for clarity. At deployment, \texttt{HEAT} fixes each site to its learned count $i^{s,\star}$, so the selected iteration depths can be interpreted as a learned allocation of numerical precision across the network. As in \texttt{QAT}, the weights adapt to the numerical error encountered at deployment; unlike conventional quantization, \texttt{HEAT} optimizes multiplicative depth rather than bit width.
\section{Experiments}\label{sec:experiments}

When compiling a network for an \texttt{FHE} pipeline, iterative approximations must be introduced to compensate for the limited (native) operation set supported by encrypted computation (\Cref{app:approx-operators}), and any additional compute required
% arithmetic operations under \texttt{CKKS} is approximated and 
further slows and weighs down the system due to increased bootstraps and memory footprint. At our deployed parameters (\Cref{app:crypto}), a single bootstrap costs two to three orders of magnitude more than any other \texttt{CKKS}~\citep{CheonHomomorphicNumbers, CheonAEncryption} primitive (\Cref{app:tab:primitives}), raising two questions: \emph{How fast does encrypted inference run?} \emph{How well is the original model's performance preserved?}

These questions guide the design of our experimental setup, where we compare a per-site calibrated \texttt{GPT-2} ($124$M) against our \texttt{HEAT} fine-tuned model. All latency and fidelity results of this section are measured under deployed end-to-end encrypted inference (thus accounting for any integration overhead) with \texttt{CKKS} at $128$-bit security on a \texttt{FIDESlib}-based~\citep{agullodomingo2025fideslibfullyfledgedopensourcefhe} backend. Packing follows \texttt{Cachemir}~\citep{yu2026cachemir}, nonlinearities use \texttt{THOR}-style approximations (\cref{app:approx-operators}, \citealt{JunghoTHOR}), $\arg\max$ follows \citet{avitan2025efficientdecodingmethodslanguage}, and the bootstrap schedule is produced by a single shared algorithm for every set of weights. Every experiment runs on a single custom \texttt{A100-64GB}, 8 cores of an Intel Xeon Platinum 8358 and \texttt{128GB} of \texttt{RAM} (\Cref{app:impl}).
\begin{figure}
    \centering
    \begin{subfigure}[t]{\figpair}
      \centering
      \includegraphics[width=\linewidth]{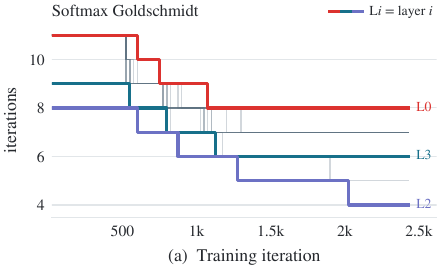}
      \phantomsubcaption\label{fig:softmax-gs-count}
    \end{subfigure}\hfill
    \begin{subfigure}[t]{\figpair}
      \centering
      \includegraphics[width=\linewidth]{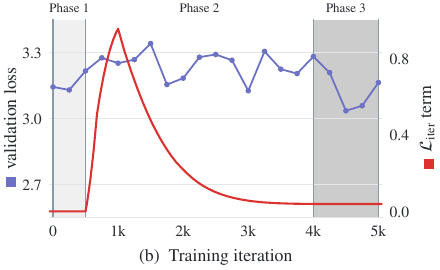}
      \phantomsubcaption\label{fig:training-dynamics}
    \end{subfigure}
    \vspace{-2ex}
    \caption{(a) Iteration counts of the Softmax's Goldschmidt refinement loop across layers during \texttt{HEAT}. The descent is not uniform, each layer settles at the depth its approximation demands. (b) \texttt{HEAT} fine-tuning dynamics: the validation loss stays flat while the per-token iteration budget $\mathcal{L}_\textrm{iter}$ decreases. The $\mathcal{L}_\textrm{iter}$ terms rises as the prior's success probability increases $0.6\!\to\!0.9$ and decays as the distributions follow it. Validation stays roughly constant ($3.0$--$3.3$ band) during fine-tuning.}
    \label{fig:descent}
    \vspace{-4ex}
\end{figure}
\paragraph{Baselines.} \label{sec:calibration}
The accuracy of an encrypted model depends on how well its nonlinearities are approximated, which depends on how their iteration counts are set. Existing methods in the literature choose a single parameterization per type of nonlinearity and apply it at every occurrence in the model. Because this single parameterization must accommodate the \emph{worst-case} site, every other site pays (in computation) for precision it does not need, inflating bootstrapping cost. A per-site calibration mitigates this, and we adopt it as our calibrated baseline; each site receives the smallest iteration count that meets a global precision target $\varepsilon \in \{10^{-i}\}_{i=1,\dots,4}$ on its own input range (over a set of $128$ sequences). We report results with $\varepsilon{=}10^{-4}$, as it empirically returns the best stability-quality trade-off over long sequences. Furthermore, we compare with \texttt{ATLAS}~\citep{xie2026atlasautomatedapproximationtransformers}, which replaces per-site precision target with a network-level objective; we run its search from the calibrated circuit on frozen weights. By contrast, \texttt{HEAT} fine-tunes the model to acquire that tolerance.

\paragraph{\texttt{HEAT} fine-tuning.} We fine-tune with \texttt{HEAT} a \texttt{GPT-2} initially calibrated at $\varepsilon{=}10^{-4}$. The three phases take $7$ \textrm{GPU}-hours in total, processing ${\approx}3.5\times10^{7}$ OpenWebText's tokens (under $1\%$ of \texttt{GPT-2}'s pretraining corpus), thanks to the forward-mode depth gradients described in \Cref{app:mem-eff}, training with iterative approximations induces negligible memory-compute costs. The learnable distributions are initialized as detailed in \Cref{app:init} and the network is trained with AdamW~\citep{adamw}; validation is the mean over $32$ held-out batches every $250$ updates. The full configuration is available in the released code.
\begin{wrapfigure}{r}{0.50\textwidth}
\centering
    \includegraphics[width=0.49\textwidth]{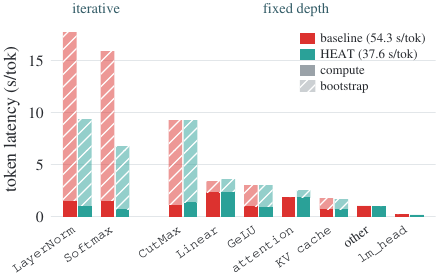}
    \caption{Per-family decode-latency composition, calibrated \texttt{GPT-2} vs.\ \texttt{HEAT}: paired bars in absolute seconds per token, the hatched cap of each bar being the bootstrap time spent inside that family. In the baseline, the iterative nonlinearities own $59\%$ of the token wall and $71\%$ of all bootstrap time; \texttt{HEAT}'s saving concentrates in exactly those bars. \texttt{CutMax} is excluded from optimization. 
    % Where computational time varies is to attribute to different circuits that generate different landing levels, that have different execution times.
    }
    \label{fig:histo-latency-heat}
    \vspace{-2ex}
\end{wrapfigure}
\Cref{fig:softmax-gs-count} shows how iteration counts evolve during fine-tuning, taking Goldschmidt's iterations in the Softmax as an example. Depth is learned layer by layer, settling at different values. \texttt{HEAT} lowers precision only where the weights can absorb it rather than compressing every site uniformly; under the \texttt{HEAT}-\texttt{QAT} analogy, this is a learned mixed-precision allocation. Summed over the network, the learned circuit runs $228$ solver iterations per forward against the calibrated $712$, a $3.1\times$ reduction (\Cref{app:fig:depth-map}). \Cref{fig:training-dynamics} traces the corresponding loss. The $\mathcal{L}_{\textrm{iter}}$ term is held off during phase one warm-up, and it ramps in as the prior anneals to put progressively stronger drag on the counts; throughout, the validation loss stays within the $3.0$--$3.3$ band, \emph{i.e.}\ the model releases iterations without bartering perplexity for them, thanks to co-adaptation. 

Fewer iterations consume fewer levels per forward pass, and fewer levels mean fewer bootstraps; \Cref{fig:histo-latency-heat} breaks down the decode latency by operator family. In the calibrated baseline, the two iterative nonlinearities, LayerNorm and Softmax, account for $59\%$ of the per-token wall time and $71\%$ of all bootstrap time. \texttt{HEAT}'s saving lands exactly on those two bars, roughly halving each, while the fixed-depth families (linear layers, \textrm{GELU}, attention, \textrm{KV} cache) are left untouched. Since bootstrapping inside the iterative methods dominates the encrypted forward pass, this is the cost \texttt{HEAT} is designed to reduce.
\begin{table}
    \centering
    \small
    \setlength{\tabcolsep}{3pt}
    \caption{Encrypted decode at deployment. \emph{it/fwd}: iterations of the iterative approximations per forward; \emph{bts/tok} and \emph{s/token}: bootstraps and wall time per generated token; \emph{fidelity} over $64$ paired OpenWebText sequences of $128$ teacher-forced positions, each model against its own plaintext logits. \emph{Coll.}\ denotes collapsed sequences, \emph{i.e.}\ those in which more than half of the final states diverged during encrypted evaluation. \emph{Top-1} is computed over all sequences, tokens that fail to decrypt counting as misses. \emph{Perplexity} is computed over stable sequences and measured over the $128$ tokens; subscript is the perplexity of the same circuit evaluated in plaintext. $^{\star}$Uses ring dimension $N{=}2^{17}$ and $20$ usable levels, not reproduced by us. $^{\dagger}$Search run from the calibrated circuit; cheapest point within $2\%$ of its plaintext perplexity; its gap under encryption is an encrypted effect (\Cref{app:drift}).}
  \begin{tabular}{@{}>{\columncolor{white}[0pt][\tabcolsep]}l cc cc c c>{\columncolor{white}[\tabcolsep][0pt]}c@{}}
    \toprule
    & \multicolumn{2}{c}{circuit} & \multicolumn{2}{c}{s/token} & \multicolumn{2}{c}{fidelity} & perplexity \\
    \cmidrule(lr){2-3}\cmidrule(lr){4-5}\cmidrule(lr){6-7}\cmidrule(lr){8-8}
    \textbf{System} ($\lambda{=}128$) & it/fwd\dn& bts/tok\dn & Transf.\dn & e2e\dn & coll.\dn & top-1\up & $\texttt{FHE}_{(\textrm{plain})}$\dn \\
    \midrule
    \rowcolor{gray!15}
    \,\texttt{GPT-2} (plaintext) & -- & -- & -- & -- & -- & -- & 32.4 \\
    \cdashlinelr{1-8}
    \,\citet{Castro2025EncryptedLLM:Encryption}$^{\star}$ & -- & -- & 60.0 & ${\approx}68$ & -- & -- & -- \\
    \cdashlinelr{1-8}
    \,\texttt{GPT-2} (calibrated $10^{-4}$) & 712 & 513 & 45.0$_{(\pm0.3)}$ & 54.3$_{(\pm0.4)}$ & 7/64 & 70.3$_{(\pm2.6)}$ & 44.6$_{(31.7)}$ \\
    \,\texttt{ATLAS}~\citep{xie2026atlasautomatedapproximationtransformers}$^{\dagger}$ & 438 & 396 & 34.5$_{(\pm0.3)}$ & 43.8$_{(\pm0.4)}$ & 13/64 & 58.1$_{(\pm3.0)}$ & 50.5$_{(33.1)}$ \\
    \,\textbf{\texttt{HEAT} (ours)} & \textbf{228} & \textbf{326} & \textbf{28.2}$_{(\pm0.2)}$ & \textbf{37.6}$_{(\pm0.3)}$ & \textbf{0/64} & \textbf{82.9}$_{(\pm0.5)}$ & \textbf{30.7}$_{(31.0)}$ \\
    \bottomrule
    \end{tabular}
    \label{tab:main-results}
    \vspace{-3ex}
\end{table}
\Cref{tab:main-results} reports the main latency-quality results. \texttt{HEAT} executes $326$ bootstraps per token against the baseline's $513$ ($1.6\times$ reduction), and serves a token in $37.6\stok$ against the calibrated baseline's $54.3\stok$ ($1.4\times$ reduction) and \texttt{EncryptedLLM}'s~\citep{Castro2025EncryptedLLM:Encryption} ${\approx}68\stok$, whose iteration counts are fixed and uniform across the model. The encrypted $\arg\max$'s~\citep{avitan2025efficientdecodingmethodslanguage} calibration is shared by every approach and contributes the same $9.3\stok$ to each row.
% \Cref{tab:vit-results} repeats the protocol on a \texttt{ViT} for encrypted \texttt{EuroSAT} classification.

Latency is half the deployment question; the other half is whether the shallower circuit remains faithful to its plaintext analog. \Cref{tab:main-results} gives the perplexity of the encrypted inference; here we study in detail how degradation happens along the sequence. \Cref{fig:results} answers this over $64$ teacher-forced chains of $128$ tokens, \emph{i.e.} we test model preservation and stability under encryption. 
\texttt{ATLAS} is near-lossless in plaintext yet collapses on $13$ of $64$ sequences under encryption. This is an encrypted effect that a noise-free search objective cannot see (\Cref{app:drift}).
Along the sequence, the calibrated baseline starts near-exact and degrades. The median \textrm{KL} grows $395\times$ from the first $32$ positions to the last ($0.0013\!\to\!0.519$) and top-$1$ agreement falls from $91\%$ to $47\%$. \texttt{HEAT} starts from a slightly higher floor and grows slower ($0.016\!\to\!0.129$; top-$1$ $90\%\!\to\!77\%$). Pooled over all decoded tokens, top-$1$ agreement is $70.3\%$ for the baseline and $82.9\%$ for \texttt{HEAT}, and over the last $32$ positions \texttt{HEAT}'s \textrm{KL} is $4\times$ lower. The shallower circuit is thus not only faster but more faithful to its plaintext analog than the calibrated one is to its own, thanks to the approximation tolerance it learned and less overall bootstrap noise injection.
\begin{figure}
    \centering
    \includegraphics[width=\linewidth]{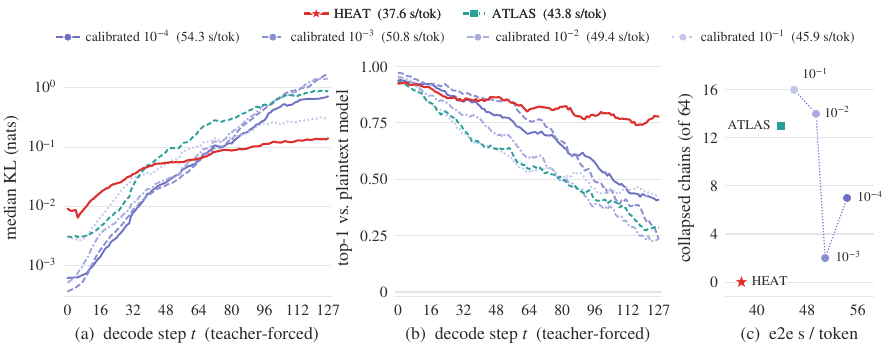}
    \caption{Per-position fidelity of each encrypted circuit against its own plaintext logits over $64$ sequences of $128$ teacher-forced decode steps. (a) Median encrypted \texttt{KL} over the sequences that do not collapse. The calibrated circuits start near-exact and degrade with position, the looser tolerances faster; \texttt{HEAT} starts higher and stays almost flat. (b) Top-1 agreement over all sequences, the calibrated ladder falls to $23$--$42\%$ by $t{=}128$ while \texttt{HEAT} ends at $78\%$. Per-step latency is position-flat for every circuit under the \texttt{Cachemir} layout. (c) Collapsed sequences of $64$ per method, \emph{i.e.}\ sequences where more than $50\%$ of the tokens couldn't be decrypted.}
    \label{fig:results}
    \vspace{-4ex}
\end{figure}
\newcommand{\lrange}{ \circled{1}\xspace}
\newcommand{\depth}{ \circled{2}\xspace}
\newcommand{\counts}{ \circled{3}\xspace}

\paragraph{Ablation study.}\Cref{tab:ablation} separates \texttt{HEAT}'s two losses, and explores two single-component alternatives:\lrange trains the weights with the range loss \emph{$\mathcal{L}_\textrm{range}$ only} and lets the calibrator set the counts afterwards;\depth trains the halting logits on frozen pretrained weights.\counts is the artisan's work, fixing every count by hand and fine-tuning the weights ($397$ is the shallowest circuit this recipe reaches without producing a circuit that always collapses as a starting point). Fine-tuning under $\mathcal{L}_\text{range}$ alone reduces only part of the latency by narrowing the approximation intervals and thus the iteration count needed to meet the precision target, but does not reach the co-adapted circuit \texttt{HEAT} learns. Conversely, transplanting \texttt{HEAT}'s learned counts onto the\lrange \emph{squeezed} and calibrated checkpoints breaks both in plaintext.
% The two collapsing options (\depth learned depths and\counts fixed counts) suffer from the same instability. During fine-tuning, the iteration counts become so small that the error of one approximation pushes the input of the next outside its convergence interval, and the outputs explode. Once a batch produces non-finite samples, it contributes no gradient, so most updates are nullified; the parameters cannot move, and there is no way back out of such a configuration from the task loss only. This is what motivates \texttt{HEAT}'s gradual descent, the prior lowers the counts slowly enough that the weights adapt before any site leaves its convergence domain, and the range loss shrinks the approximations' operative margins so that fewer samples collapse in the first place. It is this combination that lets \texttt{HEAT} decrease the learned depths stably where either component alone cannot.

The two collapsing options,\depth and\counts, fail for different reasons. In\depth the counts descend until samples diverge and carry no gradient, so the task loss offers no way back. The artisan's\counts works in plaintext with counts cut to what it just tolerates; under encryption every bootstrap adds error to every slot, compounding along the sequence until one approximation pushes the next input outside its convergence interval, the outputs explode and the sequence is lost. This is what motivates \texttt{HEAT}'s gradual descent: the prior lowers the counts slowly enough that the weights adapt before any site leaves its convergence domain, and the range loss shrinks the approximations' operative margins so that fewer samples collapse in the first place. It is this combination that lets \texttt{HEAT} decrease the learned depths stably where either component alone cannot.
\begin{table}[h]
\centering
      \small
      \setlength{\tabcolsep}{5pt}
      \caption{Ablation of \texttt{HEAT}'s three components. The range loss alone gives a circuit that replicates its plaintext model but stays deep; learned depth alone, on frozen weights, cuts the iterations by a third (\emph{vs.} the calibrate baseline) but blows up in $62$ of $64$ samples, the weights never having adapted to the shallow circuit. Only the combination is both shallow and stable.}
      \begin{tabular}{@{}l ccc ccc c@{}}
      \toprule
      & $\mathcal{L}_\textrm{range}$ & $\mathcal{L}_\textrm{iter}$ & trained weights & it/fwd\dn & e2e\dn & coll.\dn & perplexity $\texttt{FHE}_{(\textrm{plain})}$ \\
      \midrule
      % \,\texttt{GPT-2} (calibrated $10^{-4}$) & \no  & \no  & \no  & 712 & 54.3 & 7/64  & 44.6$_{(31.7)}$ \\
      % \cdashlinelr{1-8}
      \lrange  squeezed     & \yes & \no  & \yes & 455 & 43.6 & 5/64  & 32.6$_{(32.5)}$ \\
      \depth learned depth only                    & \no  & \yes & \no  & 456 & 46.6 & 62/64 & -- \\
      \counts fixed counts                          & \no  & \no  & \yes & 397 & 43.1 & 58/64 & -- \\
      \cdashlinelr{1-8}
      \,\textbf{\texttt{HEAT} (ours)}         & \yes & \yes & \yes & \textbf{228} & \textbf{37.6} & \textbf{0/64} & \textbf{30.7}$_{(31.0)}$ \\
      \bottomrule
      \end{tabular}
    \label{tab:ablation}
    \vspace{-3ex}
\end{table}
\section{Related Work}
\textbf{\texttt{FHE} frameworks.} \;\;
Private transformer inference is either interactive, with protocols that split the computation between client and server~\citep{Knott2021CrypTen:Learning,HaoIron:Transformers,Li2022MPCFormer:MPC,PangBOLT:Transformers,GuptaSigma:Sharing,LuBumbleBee:Transformers,HouCipherGPT:Inference,kei2025shaft,Dong2023PUMA:Minutes}, or non-interactive under \texttt{FHE}, where the server evaluates the whole model on ciphertext. The latter builds on \texttt{CKKS}~\citep{CheonHomomorphicNumbers,CheonAEncryption} and bootstrapping~\citep{gentry2009fully}; \citet{al2026sok} systematize existing frameworks along a packing axis for linear blocks and a model-preservation axis for nonlinear ones. On the packing axis, \texttt{NEXUS}~\citep{ZhangSecureNon-interactive}, \texttt{MOAI}~\citep{moai}, and \texttt{STIP}~\citep{2026/174stip} cut rotations through column, diagonal, and complex-slot layouts; \texttt{EncryptedLLM}~\citep{Castro2025EncryptedLLM:Encryption} and \texttt{Cachemir}~\citep{yu2026cachemir} target \textrm{GPU}-accelerated decoding, the latter with an encrypted \textrm{KV} cache and the \texttt{THOR}~\citep{JunghoTHOR} nonlinear circuits we also adopt; \textrm{GPU} libraries~\citep{agullodomingo2025fideslibfullyfledgedopensourcefhe} provide the primitives, and \texttt{FHE} compilers~\citep{Viand_2021,FhelipeCompiler,OrionCompiler,cheon2024dacapo} automate packing and bootstrap placement for a given circuit, though most surveyed frameworks still place bootstraps by hand. All of these take the pretrained model as given; \texttt{HEAT} is a fine-tuning that yields a model tolerant of a shallower circuit, and deploys on any of them.

\textbf{\texttt{FHE}-friendly approximations.}\;\;
Since \texttt{CKKS} exposes only additions, multiplications, and rotations, nonlinearities must be rebuilt from them. \citet{al2026sok} split frameworks into \emph{model-preserving} ones, which approximate the original layers, and \emph{model-modifying} ones, which substitute them and retrain. In the first class, division and inverse square root use Goldschmidt~\citep{goldschmidt_1964} and Newton iterations~\citep{Qu2022ImprovementsRoot}, Softmax uses a scaled exponential with an iterative reciprocal~\citep{LeeHETAL:Encryption,JunghoTHOR}, \textrm{GELU} a polynomial composition~\citep{JunghoTHOR}, and the decoding $\arg\max$ iterative sharpening~\citep{avitan2025efficientdecodingmethodslanguage}. In the second, networks are trained with polynomial activations~\citep{Dowlin2016CryptoNets:Accuracy,baruch2023training,zimerman2023converting,ling2025peregrine}, ReLU counts are pruned or searched~\citep{ghodsi2020cryptonas,jha2021deepreduce}, Softmax is replaced by a Gaussian kernel~\citep{rho2025encryptionfriendly} or a power function~\citep{powersoftmax,park2025powerformer}, and normalization is dropped altogether~\citep{Zhu2025TransformersNormalization}.
% Quantization-aware training has likewise been used to shrink integer widths for \texttt{BFV} and \texttt{TFHE}~\citep{legiest2023neural,StoianDeepTFHE} and \texttt{MPC}~\citep{Wu2024Ditto:MPC}. \texttt{HEAT} keeps the model-preserving circuits and the architecture of the first class, but fine-tunes the weights under them as the second does, and reaches plaintext-level fidelity where retraining-based substitutions do not.

\textbf{FHE iterative approximations reduction.}\;\;
Most systems fix one iteration count or polynomial degree per operator type, hand-tuned for the hardest site. \texttt{SLOTHE}~\citep{SLOTHE} approximates only the non-arithmetic sub-parts of each function under a tunable accuracy-latency cost model, \citet{LEE2026134016} choose per-layer polynomial degrees by dynamic programming over each layer's estimated impact on accuracy, \texttt{ELLMo}~\citep{ELLMo} redesigns Softmax and LayerNorm to cut their depth, and \texttt{ATLAS}~\citep{xie2026atlasautomatedapproximationtransformers} searches per-layer iteration counts and degrees against a network-level objective. All operate on frozen weights and can only harvest the tolerance the pretrained model already has.

\section{Conclusions}
Encrypted inference pays for precision it does not need, as every iterative approximation is usually set to a fixed count, and with the weights frozen, that count can only be lowered as far as the pretrained model already tolerates. \texttt{HEAT} lifts this limit by making the counts learnable and training the weights under the approximations they will meet at deployment. On encrypted \texttt{GPT-2} decoding, the learned circuit runs 3.1$\times$ fewer solver iterations, 1.6$\times$ fewer bootstraps, and serves a token in 37.6\,s against 54.3\,s, while its encrypted perplexity stays within 0.3 of plaintext, where the calibrated circuit drifts substantially. Neither learned depth nor weight adaptation can reach this alone. \Cref{prop:iters-are-bits} and \Cref{thm:heat-is-qat} relate iteration count to an effective bit width and bound, at fixed weights, the gap to mixed-precision \texttt{QAT}; this frames learning the counts as a precision allocation and motivates, without guaranteeing, that fine-tuning absorbs the error. The result is a model architecturally identical to its parent that any \texttt{CKKS} pipeline could deploy.

\section*{Acknowledgments}
This work is supported through the MUR FIS2 grant n. FIS-2023-00942 “NEXUS” (cup B53C25001030001). This work was supported by NSF award \#2439951.

\section*{Use of AI}
Large language models were used extensively to write and refactor the training code, the \texttt{CKKS} backend integration, and the Lean~4 formalization. All code was reviewed and run by the authors, every proof was checked by the Lean kernel, and the method, experiments, and analysis are the authors' own, who take full responsibility for the content.

\section*{Ethical Statement}
Lowering the cost of encrypted inference lets users in privacy-sensitive settings query hosted models without disclosing their data. The same property impacts the provider's ability to inspect prompts or outputs, so server-side content filtering, abuse detection, and audit cannot be applied, and a faster encrypted pipeline serves bad actors as readily as hospitals. \texttt{HEAT} does not create this trade-off, which is intrinsic to \texttt{FHE}, but brings it closer to practice. Our threat model also does not hide sequence length or protect the weights; \Cref{app:societal} discusses these points further.

\section*{Reproducibility Statement}
Code, Lean~4 proofs, and all configurations are at \url{https://github.com/gladia-research-group/heat}; weights are at \url{https://huggingface.co/gladia/heat-gpt2-small-openwebtext}. Cryptographic parameters, initialization, and schedule are in \Cref{app:impl}, the circuits in \Cref{app:approx-operators}. Fine-tuning takes 7 \textrm{GPU}-hours on one \texttt{A100}; every latency and encrypted-fidelity number (\Cref{tab:main-results,tab:ablation,tab:vit-results,tab:tolerance-ladder}, \Cref{fig:histo-latency-heat,fig:results}) comes from a deployed \texttt{CKKS} run on our modified \texttt{FIDESlib} build (OpenFHE 1.4.2, CUDA 12.6); the sensitivity, downstream, drift and scaling analyses of \Cref{app:additional-results} are plaintext evaluations of the deployed circuits and are labeled as such.

%\onecolumn
\bibliographystyle{iclr2027_conference}
\bibliography{references}

\newpage
\appendix
\appendixthmnumbering
\crefalias{section}{appendix}
\crefalias{subsection}{appendix}
\crefalias{subsubsection}{appendix}
\appendixthmnumbering
\section{Limitations}
\texttt{HEAT} acts only on the iterative sites; the polynomial \textrm{GELU}, the linear families, and the encrypted $\arg\max$ (a fixed $9.3\stok$) are left at their calibrated depth. We evaluate one language model and one vision transformer on a single \texttt{A100}; whether the gains carry to billion-parameter models or other approximation families is untested. Finally, \texttt{HEAT} needs gradient access and $3.5\times10^{7}$ tokens of fine-tuning.

\section{Proofs: (F)HE-Aware Training and Quantization-Aware Training}\label{app:proofs}
A machine-checked Lean~4 formalization of the proofs in this appendix is provided at \url{https://github.com/gladia-research-group/heat}.

\paragraph{Purpose.} \texttt{HEAT} fine-tunes a network while its nonlinearities are replaced by iterative approximants. On its own this is a new training regime, and nothing says a priori that a network can absorb the error of an iterative method, or that gradient descent should find weights that do. The results of this appendix do not answer that question, but they relate it to a known one, by \emph{analogy} rather than by a new theory.

\Cref{prop:iters-are-bits} shows that the worst-case error of an $n$-step quadratically convergent iterator is that of a uniform quantizer at $b^\star(n)=\Theta(2^n)$ bits, so the iteration count is a precision knob with a known exchange rate to bits. 

\Cref{thm:heat-is-qat} shows that, at fixed weights and matched precision, a network evaluated with iterative approximants and the same network evaluated with quantized weights differ in task loss by an amount that vanishes doubly exponentially in $n$, \emph{i.e.}\ the two error-injection models have per-layer errors within the same budget. 

Training under iterative approximations can therefore be read as quantization-aware training in the currency of multiplicative depth, and learning a per-site iteration count as mixed-precision allocation. What \texttt{HEAT} inherits from this analogy is the framing of \texttt{QAT}, which trains networks to tolerate a bounded per-layer perturbation of matched magnitude, rather than a new hypothesis about what networks tolerate.

\paragraph{Iterative approximators.}
Fix constants in $\mathbb{R}$,  $M>0$, $r\in(0,1)$ and, for $n\in\mathbb{N}$, write
\begin{equation}\label{eq:def-eps}
    \varepsilon_n \;=\; Mr^{2^n}
\end{equation}
to be an iterator's error budget. As defined in \Cref{eq:iterative-convergence}, the $n$-step iterator $F_n$ approximates the nonlinearity $f$ with
\begin{equation}\label{eq:iter-conv}
    |F_n(x)-f(x)|\;\le\;\varepsilon_n\quad\text{for all }x\in\mathcal{D},
\end{equation}
where $\mathcal{D}$ is the set of inputs such that the iterator's seed strictly lies in the convergence interval. Since $0<r<1$, the sequence $\varepsilon_n$ is strictly decreasing and $\varepsilon_n\to0$ doubly exponentially.

\paragraph{Uniform quantizer.} 
A uniform quantizer $Q$ rounds to the nearest of $2^b$ uniformly spaced buckets on an interval $I$ of width $R>0$, at bit width $b\in\mathbb{N}_{\ge1}$ and with step $s(b)=R/(2^b-1)$. For $w\in I$ it holds $|w-Q(w)|\le s/2$; outside $I$ the quantizer saturates and the bound fails, so we assume no overload, \emph{i.e.}\ every entry of every $W_\ell$ lies in $I$. This is always possible with $I=[-\kappa,\kappa]$, since the row-norm bound below gives $|W_{ij}|\le\kappa$. Since $s(b)$ is decreasing in $b$, so is the worst-case rounding error. In classic neural network quantization,  only weights are quantized; activations stay in full precision, so rounding reaches them only through the linear maps (\Cref{lem:act-perturb}).

\paragraph{Network's assumptions and definitions.} We define a layer as a linear map $\mathbb{R}^d \rightarrow \mathbb{R}^{d}$, followed by a nonlinear function applied component wise. Writing $f(u)\in\mathbb{R}^d$ for the entrywise application of $f\colon\mathbb{R}\to\mathbb{R}$ to $u\in\mathbb{R}^d$, the layer with weights $W\in\mathbb{R}^{d\times d}$ and nonlinearity $f$ is the map $h\mapsto f(Wh)$. Activations live in $(\mathbb{R}^d,\|\cdot\|_\infty)$. Given per-layer weights $W_1,\dots,W_L$ and an input $h_0$, we define the layers notation for the two possible error injection models discussed in (\Cref{sec:approx-quant}) the quantization error $\cdot^Q$ and the approximation error  $\cdot^H$:
\[
    h_\ell=f(W_\ell\,h_{\ell-1}), \qquad h^{\mathrm{H}}_\ell=F_n(W_\ell\,h^{\mathrm{H}}_{\ell-1}), \qquad h^{\mathrm{Q}}_\ell=f\bigl(Q(W_\ell)\,h^{\mathrm{Q}}_{\ell-1}\bigr),
\]
Respectively the exact network, the \texttt{HEAT} network at fixed iteration count $n$, and the \texttt{QAT} network at bit width $b$.

We assume $f$ is $\lambda$-Lipschitz, with $\lambda>0$, and every weight matrix has rows of $\ell_1$-norm at most $\kappa$, \emph{i.e.}\ $\max_i\sum_j|W_{ij}|\le\kappa$. Together these make each layer $\rho$-Lipschitz with $\rho=\lambda\kappa$.
We further assume the activations of the \emph{quantized} network satisfy $\|h^{\mathrm{Q}}_\ell\|_\infty\le B$ for all $\ell\le L$ with $B>0$, that the loss is $\xi_{\mathrm{lip}}$-Lipschitz in the final activations, and that every component of the pre-activation of the \texttt{HEAT} network, $(W_\ell\,h^{\mathrm{H}}_{\ell-1})_i$, lies in $\mathcal{D}$, so that \eqref{eq:iter-conv} applies to it component wise. This last assumption is a hypothesis on the input. Calibration in practice does not imply it for every sample, and \Cref{tab:ablation} shows the collapse that follows when it fails.

\subsection{Proof of \texorpdfstring{\Cref{prop:iters-are-bits}}{Proposition~\ref{prop:iters-are-bits}}}\label{app:resolution}
We restate the \cref{prop:iters-are-bits}:
\begin{proposition}
Let $R>0$ be the quantizer range width, let $M>0$ and $r\in(0,1)$ be the real valued constants of \eqref{eq:iter-conv}, and let $n\in\mathbb{N}$. Then $b$ satisfies $s/2=\varepsilon_n$, with $s=R/(2^b-1)$, if and only if
\[
    b \;=\; b_n \;:=\; \log_2\Bigl(1+\tfrac{R}{2Mr^{2^n}}\Bigr),
\]
Consequently $b_{n+1}/b_n\to 2$ as $n\to\infty$.
\end{proposition}
\begin{proof}
The quantizer's worst-case error is $s/2=R/\bigl(2(2^b-1)\bigr)$; the iterator's is $\varepsilon_n=Mr^{2^n}$ by \eqref{eq:iter-conv}. Equating the two and solving for $b$,
\[
\frac{R}{2(2^b-1)}=Mr^{2^n} \;\iff\; 2^b=1+\frac{R}{2Mr^{2^n}} \;\iff\; b=\log_2\!\Bigl(1+\frac{R}{2Mr^{2^n}}\Bigr),
\]
For the decomposition, write 
\[
1+\tfrac{R}{2\varepsilon_n}=\tfrac{R}{2\varepsilon_n}\bigl(1+\tfrac{2\varepsilon_n}{R}\bigr)
\]
and take logarithms:
\[
b_n=\log_2\tfrac{R}{2\varepsilon_n}+\log_2\Bigl(1+\tfrac{2\varepsilon_n}{R}\Bigr)=2^n\log_2\tfrac1r+\log_2\tfrac{R}{2M}+\log_2\Bigl(1+\tfrac{2\varepsilon_n}{R}\Bigr)
\]
Since $\ln(1+x)\le x$
\[
\leq2^n\log_2\tfrac1r+\log_2\tfrac{R}{2M}+\tfrac{2\varepsilon_n}{R\ln 2}
\]
Dividing by $2^n$, the last two terms vanish as $n\to\infty$, so $b_n/2^n\to\log_2\tfrac1r>0$; hence $b_n=\Theta(2^n)$ and $b_{n+1}/b_n\to 2$.
\end{proof}

\subsection{Proof of \texorpdfstring{\Cref{thm:heat-is-qat}}{Theorem~\ref{thm:heat-is-qat}}}\label{app:heat-is-qat}
Fix an input, and recall the three trajectories $h_\ell$, $h^{\mathrm{H}}_\ell$, $h^{\mathrm{Q}}_\ell$ of the definitions, and write
\[
\delta^{\mathrm{H}}_\ell=\|h^{\mathrm{H}}_\ell-h_\ell\|_\infty, \qquad \delta^{\mathrm{Q}}_\ell=\|h^{\mathrm{Q}}_\ell-h_\ell\|_\infty.
\]
the $\ell$-th layer difference with the non-approximated model. By \Cref{lem:act-perturb}, weight rounding at bit width $b$ perturbs each output coordinate by at most $\lambda\,\tfrac{s(b)}{2}\,d\,B$. We fix $b$ at the \emph{matched bit width}
\[
b^\star(n)\;:=\;\min\bigl\{b\in\mathbb{N}_{\ge1}:\lambda\,\tfrac{s(b)}{2}\,d\,B\le\varepsilon_n\bigr\},
\]
the smallest integer bit width $b\ge1$ at which the quantizer's per-layer error does not exceed the $n$-step iterator's. Since $\lambda\tfrac{s(b)}{2}dB$ is the worst-case error of a uniform quantizer over a range of width $\lambda dBR$, which is positive by the assumptions above, \Cref{prop:iters-are-bits} applied with $R$ replaced by $\lambda dBR$ gives $b^\star(n)=\lceil b_n(\lambda dBR)\rceil\ge1$, as $b_n>0$. Here $b_n$ is a real-valued effective precision and $b^\star(n)$ its integer ceiling, so the rounding error at $b^\star(n)$ is at most, not equal to, $\varepsilon_n$. The replacement changes only the additive constant $\log_2\tfrac{\lambda dBR}{2M}$, and $b_n\le b^\star(n)<b_n+1$, so $b^\star(n)=\Theta(2^n)$.

The proof of the theorem has four steps. First the exact layers are Lipschitz with an explicit constant (Step~0), each discretization injects a per-layer error of at most $\varepsilon_n$ relative to the exact layer (Steps~1--2), these errors accumulate identically across depth (Step~3), and the loss inherits the resulting activation bound (Step~4).
\begin{theorem}[Matched-precision objective bound,\Cref{thm:heat-is-qat}]
Let $\mathcal{L}_{\mathrm{approx}}^{(n)}$ denote the task loss of a network using $n$-step iterative approximations, and let
$\mathcal{L}_{\mathrm{QAT}}^{(b)}$ denote the task loss of the same network, with the same weights, under matched $b$-bit weight quantization. If the loss is
$\xi_{\mathrm{lip}}$-Lipschitz in the final activations and
$b=b^\star(n)$ is the matched bit width, then
\begin{equation}
    \left|
        \mathcal{L}_{\mathrm{approx}}^{(n)}
        -
        \mathcal{L}_{\mathrm{QAT}}^{(b)}
    \right|
    \leq
    2\xi_{\mathrm{lip}}Mr^{2^n}
    \sum_{j=0}^{L-1}\rho^j.
\end{equation}
The bound vanishes doubly exponentially with $n$.
\end{theorem}
\begin{proof}
\emph{Step 0.} For any $x,y$ and any weight matrix with rows of $\ell_1$-norm at most $\kappa$, coordinate wise
\[
|f((Wx)_i)-f((Wy)_i)| \;\le\;\lambda\Bigl|\sum\nolimits_j W_{ij}(x_j-y_j)\Bigr| \;\le\;\lambda\sum\nolimits_j|W_{ij}|\;\|x-y\|_\infty \;\le\;\lambda\kappa\,\|x-y\|_\infty,
\]
so each exact layer is $\rho$-Lipschitz with $\rho=\lambda\kappa$.

\emph{Step 1 (per-layer error of the iterator).} Every pre-activation of the \texttt{HEAT} network lies strictly in $\mathcal{D}$ by assumption, so applying $F_n$ instead of $f$ to it changes each output coordinate by at most $\varepsilon_n$, directly by
\eqref{eq:iter-conv}.

\emph{Step 2 (per-layer error of the quantizer).} Rounding the weights reaches the activations only indirectly:
\begin{lemma}\label{lem:act-perturb}
Let $w,q_w,h\in\mathbb{R}^d$ with $\|w-q_w\|_\infty\le s/2$ and $\|h\|_\infty\le B$. Then
\[
\bigl|f(\langle w,h\rangle)-f(\langle q_w,h\rangle)\bigr|\;\le\;\lambda\cdot\frac{s}{2}\,d\,B.
\]
\end{lemma}
\begin{proof}[Proof of Lemma B.1]
$|\langle w-q_w,\,h\rangle|\le\sum_{i=1}^{d}|w_i-q_{w,i}|\,|h_i|\le\tfrac{s}{2}\,dB$ by the triangle inequality; the $\lambda$-Lipschitz bound on $f$ gives the claim.
\end{proof}
Applied row-wise, the quantized layer differs from the exact layer, \emph{evaluated on the same $B$-bounded input}, by at most $\lambda\tfrac{s(b)}{2}dB$ in every coordinate, which at $b=b^\star(n)$ is at most $\varepsilon_n$ by definition.

\emph{Step 3 (accumulation across depth).} How the error accumulates across layer of the network can be modeled as 
\begin{lemma}[Trajectory bound]\label{lem:trajectory-bound}
If $\delta_0=0$ and $\delta_{\ell+1}\le\rho\,\delta_\ell+\varepsilon_n$ for all $\ell<L$, then $\delta_\ell\le\varepsilon_n\sum_{j=0}^{\ell-1}\rho^j$ for all $\ell\le L$.
\end{lemma}
\begin{proof}
Induction on $\ell$: the base case is the empty sum, and $\delta_{\ell+1}\le\rho\bigl(\varepsilon_n\sum_{j=0}^{\ell-1}\rho^j\bigr)+\varepsilon_n =\varepsilon_n\sum_{j=0}^{\ell}\rho^j$.
\end{proof}
Both trajectories satisfy this recurrence with $\rho=\lambda\kappa$. For the iterative network, insert the exact layer evaluated at $h^{\mathrm{H}}_\ell$ and split by the triangle inequality:
\[
\delta^{\mathrm{H}}_{\ell+1} \;\le\; \underbrace{\bigl\|F_n(W_{\ell+1}h^{\mathrm{H}}_\ell)-f(W_{\ell+1}h^{\mathrm{H}}_\ell)\bigr\|_\infty}_{\le\,\varepsilon_n\ \text{(Step 1)}} \;+\; \underbrace{\bigl\|f(W_{\ell+1}h^{\mathrm{H}}_\ell)-f(W_{\ell+1}h_\ell)\bigr\|_\infty}_{\le\,\lambda\kappa\,\delta^{\mathrm{H}}_\ell\ \text{(Step 0)}}.
\]
For \texttt{QAT}, insert the exact layer evaluated at $h^{\mathrm{Q}}_\ell$, which is $B$-bounded by assumption:
\[
\delta^{\mathrm{Q}}_{\ell+1} \;\le\; \underbrace{\bigl\|f(Q(W_{\ell+1})h^{\mathrm{Q}}_\ell)-f(W_{\ell+1}h^{\mathrm{Q}}_\ell)\bigr\|_\infty}_{\le\,\varepsilon_n\ \text{(Step 2, matched budget)}} \;+\; \underbrace{\bigl\|f(W_{\ell+1}h^{\mathrm{Q}}_\ell)-f(W_{\ell+1}h_\ell)\bigr\|_\infty}_{\le\,\lambda\kappa\,\delta^{\mathrm{Q}}_\ell\ \text{(Step 0)}}.
\]
Since $\delta^{\mathrm{H}}_0=\delta^{\mathrm{Q}}_0=0$, \Cref{lem:trajectory-bound} yields $\;\delta^{\mathrm{H}}_L,\,\delta^{\mathrm{Q}}_L\le\varepsilon_n\sum_{j=0}^{L-1}(\lambda\kappa)^j$.

\emph{Step 4 (from activations to losses).} The loss is $\xi_{\mathrm{lip}}$-Lipschitz
in the final activations, so, passing through the exact trajectory,
\begin{align*}
|\mathcal{L}_{\mathrm{approx}}-\mathcal{L}_{\mathrm{QAT}}|
&\;\le\;\xi_{\mathrm{lip}}\,\bigl\|h^{\mathrm{H}}_L-h^{\mathrm{Q}}_L\bigr\|_\infty
\;=\;\xi_{\mathrm{lip}}\,\bigl\|(h^{\mathrm{H}}_L-h_L)+(h_L-h^{\mathrm{Q}}_L)\bigr\|_\infty\\
&\;\le\;\xi_{\mathrm{lip}}\Bigl(\bigl\|h^{\mathrm{H}}_L-h_L\bigr\|_\infty+\bigl\|h_L-h^{\mathrm{Q}}_L\bigr\|_\infty\Bigr)
\;=\;\xi_{\mathrm{lip}}\bigl(\delta^{\mathrm{H}}_L+\delta^{\mathrm{Q}}_L\bigr)\\
&\;\le\;\xi_{\mathrm{lip}}\Bigl(\varepsilon_n\sum_{j=0}^{L-1}(\lambda\kappa)^j+\varepsilon_n\sum_{j=0}^{L-1}(\lambda\kappa)^j\Bigr)
\;=\;2\,\xi_{\mathrm{lip}}\,Mr^{2^n}\sum_{j=0}^{L-1}(\lambda\kappa)^j ,
\end{align*}
each trajectory spends one $\varepsilon_n$-budget against the exact network and the bound vanishes doubly exponentially in $n$.
\end{proof}

\paragraph{Scope.} The bound compares task-loss \emph{values} at fixed weights and a fixed count. It does not identify the two perturbations, their gradients, or the optimization trajectories they induce, and it does not cover $\mathcal{L}_{\mathrm{iter}}$ and $\mathcal{L}_{\mathrm{range}}$. Its purpose is to show that training a network under iterative approximations can be related to a well-known field, quantization-aware training, in which training under approximation error is established practice.

\paragraph{Non-entrywise operators.} During the proof we used the fact that the layer is $\lambda\kappa$-Lipschitz (Step~0), and replacing the exact nonlinearity by its deployed approximant perturbs the layer output by at most $\varepsilon_n$ on the calibrated domain (Step~1). 

These two hypotheses do not require the nonlinearity to act entry-wise, \Cref{thm:heat-is-qat} therefore extends to the Softmax and LayerNorm blocks of \texttt{GPT-2}, with the block's own Lipschitz constant in place of $\lambda\kappa$ and its output error in place of $\varepsilon_n$ (the solver error scaled by the rest of the block, \emph{e.g.}\ by $|\gamma_i(x_i-\mu)|$ in LayerNorm), constants we do not compute. These blocks have cross entry operations like the variance computation, or the Softmax denominator division, where the iterator sits inside the operator and its error propagates through the remainder of the block before reaching the next layer.

\paragraph{Per-site counts and fidelity to the exact network.} Two extensions follow from the proof with no new argument. First, the iteration count enters only through the error $\varepsilon_n$ injected at each layer (Step~1), so if site $\ell$ runs $n_\ell$ iterations the recurrence of \Cref{lem:trajectory-bound} reads $\delta_\ell\le\rho\,\delta_{\ell-1}+\varepsilon_{n_\ell}$ and the same induction gives
\[
\delta_L\;\le\;\sum_{\ell=1}^{L}\rho^{\,L-\ell}\,\varepsilon_{n_\ell},
\]
which reduces to $\varepsilon_n\sum_{j<L}\rho^j$ when $n_\ell\equiv n$. Second, Step~3 already bounds the distance of each perturbed trajectory from the \emph{exact} one, so Step~4 applied to $\delta^{\mathrm{H}}_L$ alone yields the fidelity bound, and applied to both trajectories the per-site form of the theorem:
\begin{align}
\bigl|\mathcal{L}_{\mathrm{approx}}^{(n_1,\dots,n_L)}-\mathcal{L}_{\mathrm{exact}}\bigr|
&\;\le\;\xi_{\mathrm{lip}}\sum_{\ell=1}^{L}\rho^{\,L-\ell}\,Mr^{2^{n_\ell}},
\label{eq:fidelity}\\
\bigl|\mathcal{L}_{\mathrm{approx}}^{(n_1,\dots,n_L)}-\mathcal{L}_{\mathrm{QAT}}^{(b^\star(n_1),\dots,b^\star(n_L))}\bigr|
&\;\le\;2\,\xi_{\mathrm{lip}}\sum_{\ell=1}^{L}\rho^{\,L-\ell}\,Mr^{2^{n_\ell}} .
\label{eq:mixed-precision}
\end{align}
\Cref{eq:fidelity} bounds the \emph{analytic-network} error, exact nonlinearities against $n_\ell$-step solvers, in plaintext and at fixed weights. It is not the quantity \Cref{sec:experiments} reports. \Cref{tab:main-results} and \Cref{fig:results} compare each encrypted run with a plaintext reference on the same weights, a gap that includes \texttt{CKKS} and bootstrapping noise, which no iteration count removes; top-1 agreement additionally requires a logit margin larger than twice the activation error. \Cref{eq:mixed-precision} says that a network with site-specific iteration counts is within this bound of a mixed-precision quantized network with bit widths $b^\star(n_\ell)$, which is the sense in which learning the counts is a precision allocation.
\section{Implementation Details}\label{app:impl}

\subsection{Cryptographic settings}\label{app:crypto} All runs use the \texttt{RNS} variant of \texttt{CKKS}~\citep{CheonAEncryption} at $128$-bit security (enforced by \texttt{OpenFHE}). We instantiate a cyclotomic ring of dimension $N{=}2^{16}$ ($2^{15}$ slots) with a $60$-bit first modulus and a uniform tower layout, every rescaling prime is $53$ bits, pinning the scaling factor at $\Delta{=}2^{53}$ on every level of a bootstrappable modulus chain of $28$ \texttt{RNS} limbs, for a total multiplicative depth of $27$. A bootstrap returns ciphertexts at level $16$, leaving $11$ levels for the arithmetic between refreshes, however only $8$ are usable in practice. Key switching is hybrid with $\mathrm{dnum}{=}7$ (four special primes), and we use a dense uniform ternary secret with sparse-secret encapsulation for bootstrapping (a transient Hamming-weight-$32$ key). Bootstrapping runs a single iteration with a $(4,3)$ CoeffsToSlots/SlotsToCoeffs level budget; measured over uniform inputs of increasing magnitude, the refresh keeps the worst-case relative error below $2^{-10}$ for $|x|\in[1.38,246]$ and below $10^{-4}$ (${\approx}13$ bits) in the inner band $|x|\in[12.85,76.4]$, $10$--$13$ bits depending on the magnitude, from which we estimated the iteration calibration target.
\begin{figure}[t]
    \centering
    \includegraphics[width=\figsingle]{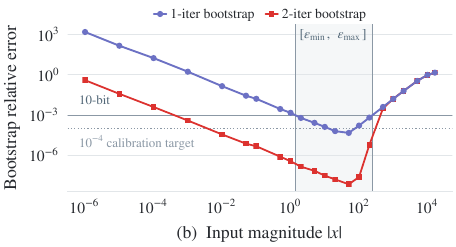}
    \caption{Bootstrap precision at different magnitude of the underlying plaintext with fixed correction factor $7$, \texttt{FIDESlib}'s default.}
    \label{fig:bts-precision}
\end{figure}
\paragraph{Primitive costs and precision.} \Cref{app:tab:primitives} is the cost model behind the claim that multiplicative depth drives latency. On a single custom \texttt{A100-64GB} \textrm{GPU}, bootstrap dominates every other primitive by two to three orders of magnitude, so a model that tolerates a shallower circuit buys its latency almost entirely through the refreshes it avoids. We time each \texttt{CKKS} primitive in isolation, $100$ repetitions each, at the deployment parameters above, on the machine used for every end-to-end run: one custom \texttt{NVIDIA A100-64GB}, an \texttt{Intel Xeon Platinum 8358} ($8$ cores and $128$\,GB of \texttt{RAM}), \texttt{CUDA}~$12.6$, \texttt{OpenFHE}~$1.4.2$, and a modified \texttt{FIDESlib} build. We also report the precision of the bootstrap operation over random sampled inputs against their magnitude \cref{fig:bts-precision}.
\begin{table}[t]
  \centering
  \small
  \setlength{\tabcolsep}{7pt}
    \caption{Each \texttt{CKKS} primitive timed in isolation on one custom \texttt{A100-64GB} at the deployment parameters of \Cref{app:crypto} above, $100$ repetitions each. \emph{levels} is the multiplicative depth the operation consumes; the bootstrap row reads $-8$ as it returns a ciphertext from the refresh trigger at level $24$ to level $16$ (due to instability of deeper levels). Arithmetic is timed at that same input level $16$.}
  \begin{tabular}{@{}l r rr@{}}
  \toprule
  \textbf{primitive} & \textbf{levels} & \textbf{mean} (ms) & \textbf{min} (ms) \\
  \midrule
  add ($\mathsf{ct}+\mathsf{ct}$)          & $0$  & $0.047$  & $0.045$  \\
  sub ($\mathsf{ct}-\mathsf{ct}$)          & $0$  & $0.047$  & $0.045$  \\
  mult ($\mathsf{ct}\times\mathsf{ptx}$)   & $1$  & $0.066$  & $0.042$  \\
  mult ($\mathsf{ct}\times\mathsf{ct}$)    & $1$  & $0.580$  & $0.559$ \\
  square ($\mathsf{ct}^2$)                 & $1$  & $0.566$  & $0.544$  \\
  rotate (one key-switch)                  & $0$  & $0.575$  & $0.565$ \\
  \midrule
  \textbf{bootstrap}                       & $\mathbf{-8}$ & $\mathbf{82.1}$ & $81.9$ \\
  \bottomrule
  \end{tabular}
  \vspace{1ex}
  \label{app:tab:primitives}
\end{table}
\subsection{Memory-efficient depth gradients} \label{app:mem-eff}
The expected output $\widetilde{f}_l(x)=\sum_i p_i^l F_i^l(x)$ of \Cref{sec:method} would, under reverse-mode autodiff, retain every intermediate solver state $F_i^l$ for the backward pass, so activation memory grows as $O(n_l)$ and the added memory traffic costs latency. Checkpointing removes the memory cost but pays for it with a recomputed forward pass \emph{and} a backward sweep over it. Since no gradient is accumulated at any intermediate $F_i^l$, we can avoid the reverse sweep. We implement each iterative family (e.g. \Cref{alg:div_forward}) as a custom autograd function that stores only the solver's inputs and re-runs the recurrence during backward, propagating forward-mode tangents alongside the recomputed iterates. The recurrences are elementwise, so the Jacobian--vector sweep costs one extra forward pass rather than a forward and a backward pass, and the gradient through the truncated iterations is exact. Activation memory is thus constant in the iteration count. We run with mixed precision; we clamp the tangent accumulators, which off-domain rows can otherwise push past the \texttt{bf16} range.

\begin{algorithm}
    \centering
    \begin{algorithmic}
        \Procedure{Step}{$N,D,F$}
    \State $N\gets N F$;\quad $D\gets D F$
    \State \Return $N,\ D,\ 2-D$
\EndProcedure
\Statex
\Procedure{Forward}{$N,D,F,\mathbf p$}
    \State $r\gets N\,\mathbf p_0$
    \For{$i=1,\dots,|\mathbf p|$}
        \State $N,D,F\gets\Call{Step}{N,D,F}$
        \State $r\gets r+N\,\mathbf p_i$
    \EndFor
    \State \Return $r$
\EndProcedure
\Statex
\Procedure{TangentStep}{$N,D,F,\ \nu,\eta,\delta,\varphi$}
    \State $\nu\gets\nu F$;\quad $\eta\gets\eta F+N\varphi$;\quad $\delta\gets\delta F+D\varphi$
    \State \Return $\nu,\ \eta,\ \delta,\ -\delta$
\EndProcedure
\Statex
\Procedure{Backward}{$N,D,F,\mathbf p,g$}
    \Comment{$g=\partial\mathcal L/\partial r$}
    \State $\nu,\eta,\delta,\varphi\gets 1,0,1,0$
    \State $\bar N_0,\ \bar D_0,\ \bar{\mathbf p}_0\gets g\,\mathbf p_0,\ 0,\ g\,N$
    \For{$i=1,\dots,|\mathbf p|$}
        \State $\nu,\eta,\delta,\varphi\gets\Call{TangentStep}{N,D,F,\nu,\eta,\delta,\varphi}$
        \Comment{pre-step $N,D,F$}
        \State $N,D,F\gets\Call{Step}{N,D,F}$
        \State $\bar N_0\gets\bar N_0+g\,\mathbf p_i\nu$;\quad $\bar D_0\gets\bar D_0+g\,\mathbf p_i\eta$
        \State $\bar{\mathbf p}_i\gets g\,N$
    \EndFor
    \State \Return $\bar N_0,\ \bar D_0,\ \bar{\mathbf p}$
\EndProcedure
    \end{algorithmic}
\caption{Forward and backward pass of our Goldschmidt's $\widetilde f_l(x)$. Bars denote adjoints, $\bar x=\partial\mathcal L/\partial x$, with $N_0,D_0$ the inputs to \textsc{Backward}. Tangents are
$\nu=\partial N/\partial N_0$, $\eta=\partial N/\partial D_0$,
$\delta=\partial D/\partial D_0$, $\varphi=\partial F/\partial D_0$.}
    \label{alg:div_forward}
\end{algorithm}

\subsection{Parameterization of the halting distribution} \Cref{sec:method} treats each site's distribution $\mathbf{p}^s$ abstractly; in the implementation it is not stored as $n{+}1$ free probabilities but in the sequential form of \citet{banino2021pondernet}, which is what keeps it a distribution without an explicit normalization. Each site carries $n$ learnable \emph{halting logits} $\ell^s_0,\ldots,\ell^s_{n-1}$; writing $h^s_i=\sigma(\ell^s_i)=\bigl(1+e^{-\ell^s_i}\bigr)^{-1}$ for the probability of stopping at iteration $i$ \emph{given} that iteration $i$ was reached,
\begin{equation}\label{eq:halting-param}
    p^s_i \;=\; h^s_i \prod_{j<i}\bigl(1-h^s_j\bigr) \quad (i<n),
    \qquad
    p^s_n \;=\; \prod_{j<n}\bigl(1-h^s_j\bigr),
\end{equation}
\emph{i.e.}\ the deepest state absorbs the remaining mass ($h^s_n\equiv 1$), so $\mathbf{p}^s$ is a distribution over $\{0,\ldots,n\}$ by construction while the $\ell^s_i$ stay unconstrained. We evaluate \Cref{eq:halting-param} in log space, $\log p^s_i=\log\sigma(\ell^s_i)+\sum_{j<i}\log\sigma(-\ell^s_j)$, and read the last atom off as the complement $p^s_n=\bigl(1-\sum_{i<n}p^s_i\bigr)_+$. This is the product in \Cref{eq:halting-param} exactly, and the clamp only absorbs floating-point round-off. Nothing in \Cref{sec:method} depends on this choice. The expected approximation, $\mathcal{L}_{\textrm{iter}}$ and the exported mode $i^{s,\star}$ are functions of $\mathbf{p}^s$ alone, and gradients reach the $\ell^s_i$ only through it. At a site with a floor $m^s>0$ (Softmax initialization and refinement, LayerNorm-Goldschmidt $m^s{=}1$; LayerNorm-Newton $m^s{=}0$) the logits $\ell^s_0,\ldots,\ell^s_{m^s-1}$ are dropped from the recursion altogether, so $\mathbf{p}^s$ is supported on $\{m^s,\ldots,n\}$: states below the floor carry exactly zero mass, receive no gradient, and the exported index is offset by $m^s$.

\subsection{Initialization and schedule}\label{app:init} Each halting distribution spans the solver's states from a per-family floor $m^s$ up to a per-family maximum count. Each distribution is initialized by a monotone ramp on the halting logits of \Cref{eq:halting-param}, $\ell^s_i=\gamma\,(i-(n-1))$ with $\gamma{=}0.7$. $h^s_i$ is thus near zero at shallow $i$ and reaches $\tfrac12$ one step below the maximum. The resulting $\mathbf{p}^s$ is monotone non-decreasing, its two deepest states carrying equal mass, so training starts at the top of each site's support and the descent is free to settle above the calibrated count where the task demands it, as LayerNorm-Newton does at blocks $0$--$1$. A flat law spiked on the deepest state, which we tried first, leaves the shallow states without gradient and slow to reduce. During training, the inputs of every approximated operator are kept inside their calibrated domains by hard clamps carrying a quadratic push-back gradient (clamp in the forward pass, penalty gradient in the backward); the clamps are a training-time guard only, and are absent during encrypted inference (as no comparison primitive exists). At export, the deployed count of every site is the mode of its halting distribution. The prior of site $s$ is a negative-binomial law on the active support $k=i-m^s$, $p^{s,\star}_k\propto\binom{k+r-1}{k}p^{r}(1-p)^{k}$, renormalized on the finite support, with shape $r=1+(w^s+\tfrac12)\,p/(1-p)$ so that its mode sits at the target $w^s$ (a geometric law when $w^s{=}0$).

\subsection{Training configuration (\texttt{GPT-2} small)}\label{app:config}
\Cref{app:tab:config} is the recipe behind every \texttt{HEAT} row of \Cref{tab:main-results,tab:ablation}. The three phases of \Cref{sec:heat-schedule} run back to back on one \textrm{GPU} at context $1024$ on OpenWebText, with AdamW~\citep{adamw} ($\beta{=}(0.9,0.99)$, $\epsilon{=}10^{-10}$, weight decay $0.1$, gradient norm clipped at $1$). Phase one trains the weights alone under $\mathcal{L}_\textrm{task}+\lambda_\textrm{range}\mathcal{L}_\textrm{range}$, with the pretrained model as a self-distillation teacher  ($\mathcal{L}_\textrm{task}$ is the \textrm{KL} divergence to the teacher at temperature $1$, with no cross-entropy term) and the Softmax evaluated exactly under its score squeeze. Phase two starts from that point, calibrated once at $\varepsilon{=}10^{-4}$ to seed every site, releases the halting distributions and ramps $\mathcal{L}_\textrm{iter}$ in over the next $500$; at the switch the weights' learning rate steps to its phase-two value and decays by cosine to the floor over the remaining budget. Phase three freezes every count at its learned mode and continues the weights for $1000$ updates at the floor with $\lambda_\textrm{iter}{=}0$. The halting logits carry their own learning rates. The guards of $\mathcal{L}_\textrm{range}$ bound the \textrm{GELU} input at $32$, the attention scores at the squeeze bound, the Softmax denominator at $7$, the LayerNorm variances at $16$, and keep the inverse-square-root input inside its calibrated domain. The halting initialization and the clamps are those of \Cref{app:init}; validation is the mean over $32$ held-out batches every $250$ updates.
\begin{table}[t]
  \centering
  \small
  \setlength{\tabcolsep}{7pt}
    \caption{\texttt{HEAT} fine-tuning recipe for \texttt{GPT-2} small, one column per phase of \Cref{sec:heat-schedule}. $\lambda_\textrm{iter}$ ramps in over the first $500$ updates of phase two. The prior lowers its mode once the learned mode has saturated on it for $75$ micro-batch steps, at least $150$ updates apart.}
  \begin{tabular}{@{}l rrr@{}}
  \toprule
  & \textbf{phase 1} & \textbf{phase 2} & \textbf{phase 3} \\
  & squeeze & \texttt{HEAT} & cool-down \\
  \midrule
  updates                  & $500$ & $3500$  & $1000$ \\
  tokens per update        & $32768$ & $4096$  & $4096$ \\
  \midrule
  weights lr               & $10^{-4}$ & $2.5{\times}10^{-4}$ & $3{\times}10^{-5}$ \\
  halting-logit lr         & --      & $2{\times}10^{-3}$ & $0$ \\
  \midrule
  $\lambda_\textrm{range}$ & $3{\times}10^{-2}$ & $0.3$ & $0.3$ \\
  $\lambda_\textrm{iter}$  & --      & $0.6$   & $0$ \\
  score squeeze bound      & $10$    & $14$    & $14$ \\
  \midrule
  halting                  & --      & soft (\cref{eq:expected-approximation}) & hard \\
  prior $p$, start / end   & --      & $0.6$ / $0.9$ & -- \\
  back-off patience / gap  & --      & $75$ / $150$ & -- \\
  self-distillation        & yes     & --      & -- \\
  \bottomrule
  \end{tabular}
  \vspace{1ex}
  \label{app:tab:config}
\end{table}

\section{The Approximated Operators}\label{app:approx-operators}

Despite the limited operation set offered by \texttt{CKKS}, there is a well-studied variety of algorithms that, using only additions, multiplications, and rotations, allow nonlinear functions to be computed. Which algorithms to choose, how to compose them, and how to initialize their recursions determine part of the encrypted inference system. A common transformer, and especially \texttt{GPT-2}, includes three main nonlinear chokepoints that must be converted to the available operations to allow computation under \texttt{CKKS}: the Softmax, the LayerNorm, and the \textrm{GELU} activation. The Softmax and \textrm{GELU} circuits follow \texttt{THOR}~\citep{JunghoTHOR}, and the inverse-square-root circuit follows \citet{Qu2022ImprovementsRoot}; we summarize them here so the paper is self-contained, and point out where the iteration counts that \texttt{HEAT} learns sit inside each circuit. The Goldschmidt reciprocal of $d>0$ starts from a seed $y_0=\alpha-\beta d$, where $(\alpha,\beta)$ is the minimax linear initializer for the calibrated interval $[d_{\min},d_{\max}]$, and iterates $y_{i+1}=y_i\,(2-d\,y_i)$; writing $d\,y_i=1-e_i$ gives $e_{i+1}=e_i^{2}$, so each iteration squares the error. The Newton inverse square root of $z$ iterates $y_{k+1}=\tfrac12\,y_k\,(3-z\,y_k^{2})$ and converges quadratically likewise. Every learned halting site of \Cref{sec:method} wraps one of these two loops.

\paragraph{Softmax.} The circuit is the two-phase exp-and-reciprocal construction of \texttt{THOR}. Let a score row lie in the calibrated range $[a,b]$ of width $M=b-a$. The row is first shifted by the plaintext constant $(a+b)/2$: under oblivious execution no encrypted row-max is available, so the classical max-subtraction is replaced by a calibration-time midpoint. Two scaling factors $(\delta_1,\delta_2)$ govern the rest. The shifted row is scaled by $1/(\delta_1\delta_2)$, placing its entries in a window of width $M/(\delta_1\delta_2)$ narrow enough for a degree-$8$ Chebyshev fit of $e^{x}$; the fit is then raised to the $\delta_1$-th power by $k=\log_2\delta_1$ ciphertext squarings. Phase one closes by normalizing, the row is summed and multiplied by the Goldschmidt reciprocal of the sum, the Softmax-initialization site. By construction the result approximates $\operatorname{Softmax}(\mathbf{x}/\delta_2)$, and since $\operatorname{Softmax}(2\mathbf{x})_i=y_i^2/\sum_j y_j^2$ for $\mathbf{y}=\operatorname{Softmax}(\mathbf{x})$, phase two applies $\log_2\delta_2$ square-and-normalize passes, each one squaring plus one Goldschmidt reciprocal of the new sum, the Softmax-refinement site, to recover $\operatorname{Softmax}(\mathbf{x})$. \citet{JunghoTHOR} select $(\delta_1,\delta_2)$ by greedy search on a training set; our calibration performs the per-block equivalent, and this refinement loop is where the learned depth stays structured across layers (\Cref{app:fig:depth-map} and \cref{fig:softmax-gs-count}).

\paragraph{LayerNorm inverse square root.} The variance $z$ enters a calibrated domain $[z_{\min},z_{\max}]$, and the cost of evaluating $\nicefrac{1}{\sqrt{z}}$ under \texttt{CKKS} is dominated by how the Newton loop is seeded. \citet{Qu2022ImprovementsRoot} proposes two options, a Taylor's polynomial and rational initializers. We instantiate the rational variant. We fit a degree-$(3,1)$ Remez rational approximation of $\nicefrac{1}{\sqrt{z}}$ on the calibration domain, its division is realized by multiplying the numerator with the Goldschmidt reciprocal of the denominator, the LayerNorm-Goldschmidt site, and the resulting quotient seeds the Newton loop, the LayerNorm-Newton site. Both counts are learned; the seeds $(\alpha,\beta)$ and the rational coefficients are calibration state, held fixed during training and re-fitted once on the final weights.
\begin{figure}[h]
  \centering
  \includegraphics[width=\textwidth]{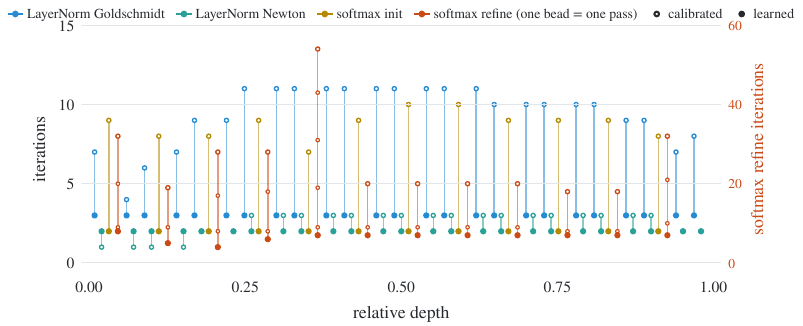}
  \caption{Per-site iteration counts, \texttt{GPT-2}: calibrated baseline (open) vs.\ \texttt{HEAT}-learned deployment (filled), one column per halting site in depth order. Family totals move $235/63/105/309 \to 75/50/24/79$ (LayerNorm-Goldschmidt/Newton, Softmax initialization/refinement), \emph{i.e.}\ $712 \to 228$ iterations per forward.}
  \label{app:fig:depth-map}
\end{figure}
\paragraph{\textrm{GELU}.} Following \texttt{THOR}~\citep{JunghoTHOR}, and unlike previous comparison-based methods~\citep{Dong2023PUMA:Minutes, Castro2025EncryptedLLM:Encryption}, the activation is a composition of two low-degree polynomials rather than a single high-degree fit, or a piecewise one. $\operatorname{GELU}(x)\approx x\,\bigl(P_2(P_1(x/S))+\tfrac12\bigr)$ with $\deg P_1=31$, $\deg P_2=27$, and $S$ a calibrated input bound, evaluated with the Paterson--Stockmeyer algorithm at deployment. The circuit is a fixed-depth polynomial with no iterative loop, which is why \textrm{GELU} carries no halting distribution and enters the training machinery only through its range guard. During \texttt{HEAT} we evaluate the analytic \textrm{GELU} and fit the polynomials directly on the fine-tuned weights.

\paragraph{Encrypted $\arg\max$.} Decoding ends with an encrypted $\arg\max$ over the vocabulary logits, following the CutMax construction of \citet{avitan2025efficientdecodingmethodslanguage}. The row is driven to a one-hot indicator of its maximum by a short sequence of \emph{sharpening} rounds. Each round re-centers and rescales the row, raises it to a small power, widening the leader's margin over the rest, and renormalizes; the normalizer is itself evaluated iteratively, by a chord-seeded Newton inverse-square-root cascade of the same kind as LayerNorm's. The deployed schedule runs five rounds (powers $9,3,3,3,7$, with one to three Newton passes each) and is derived at calibration time. Unlike the four halting families above, the $\arg\max$ carries no halting site. Its schedule is calibration state, frozen once and shared byte-identical by every approach we compare, so it contributes the same ${\approx}9.3 \stok$ to every token and none of the reported latency differences. Future work might consider including the $\arg\max$ in the optimization, despite being unstable.
\section{Additional Results}\label{app:additional-results}

\texttt{GPT-2} exposes $74$ halting sites, one per solver stage of \Cref{app:approx-operators}. LayerNorm-Goldschmidt ($25$), LayerNorm-Newton ($25$), Softmax initialization ($12$), and Softmax refinement ($12$, at $\delta_2{=}2$ for every layer after squeeze); the \textrm{GELU} is a fixed polynomial and carries no halting site. \Cref{app:fig:depth-map} compares every site's calibrated count versus what \texttt{HEAT} learns. The open marker is the calibrated baseline count, the filled marker the learned deployed count. First, we note the LayerNorm families collapse to shallow, uniform counts. Goldschmidt moves from per-site $4$--$11$ to a flat $3$, Newton from $1$--$3$ to $2$ everywhere, an \emph{increase} at blocks $0$--$1$, where the calibrator had settled on $1$. Second, Softmax refinement is the one family that allocates precision in a \emph{structured} way: block $0$ retains $8$ refinement iterations while the later blocks settle at $4$--$7$. \Cref{fig:softmax-gs-count} traces how that structure emerges during training. The layers separate early in the descent and each settles at the depth its approximation demands, rather than following a global schedule.

\subsection{Sensitivity to the seed and to the iteration penalty}\label{app:sensitivity}
We rerun the three phases of \Cref{app:tab:config} with three further seeds and, at one seed, over $\lambda_\textrm{iter}\in\{0.3,0.6,1.2\}$ crossed with the prior's end probability $p_\textrm{end}\in\{0.8,0.9,0.95\}$; everything else is unchanged. These runs are in plaintext, and their perplexity is the validation loss at context $1024$ over $160$ held-out OpenWebText batches shared by every run, not the $128$-token windows of \Cref{tab:main-results}. Every seed recovers the shipped circuit exactly, the same count at every site and hence $228$ iterations per forward, while the weights differ (perplexity $22.96\pm0.07$ over the held out samples). Over training, the four runs of the recipe, the shipped run included, step down at slightly different times but coincide from update $3325$ on, and their validation losses differ by no more than the evaluation noise (\Cref{app:fig:seed-stability}). Raising $\lambda_\textrm{iter}$ removes iterations and costs perplexity (\Cref{app:fig:lambda-sweep}). The end probability moves neither quantity by more than $3$ iterations or $0.2$ perplexity, with one exception, at $\lambda_\textrm{iter}{=}0.3$, $p_\textrm{end}{=}0.8$ barely compresses ($362$ iterations, against $240$ at $p_\textrm{end}{=}0.9$).
\begin{figure}[h]
\centering
\begin{subfigure}[t]{\figpair}
  \centering
  \includegraphics[width=\linewidth]{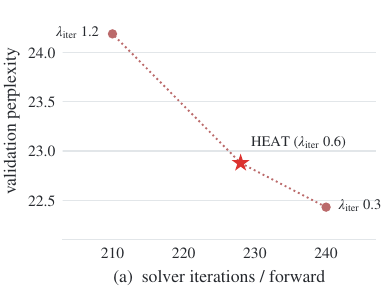}
  \phantomsubcaption\label{app:fig:lambda-sweep}
\end{subfigure}\hfill
\begin{subfigure}[t]{\figpair}
  \centering
  \includegraphics[width=\linewidth]{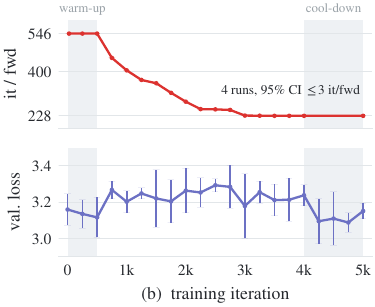}
  \phantomsubcaption\label{app:fig:seed-stability}
\end{subfigure}
\caption{(a) Raising $\lambda_\textrm{iter}$ trades solver iterations for validation perplexity (one seed, $p_\textrm{end}{=}0.9$, everything else as in \Cref{app:tab:config}); the star is the shipped setting. (b) Four runs of the shipped recipe, the shipped run and three further seeds, as mean and $95\%$ confidence interval ($n{=}4$). The iteration count follows the same descent starting from the sum of the modes of the initialized distributions; the median interval on the validation loss ($\pm0.087$) equals the one at update $0$ ($\pm0.088$), where the runs share their weights and differ only in their $32$ evaluation batches. Shaded: phase one warm-up and phase three cool-down.}
\label{app:fig:sensitivity}
\end{figure}

\subsection{Image classification task} \Cref{tab:vit-results} repeats the protocol of \Cref{tab:main-results} on encrypted \texttt{EuroSAT} classification with a \texttt{ViT-B/16} ($86$M parameters). $1024$ test images drawn at random, images whose logits do not decrypt counted as misses. \texttt{HEAT} is the fastest circuit, with $2.5\times$ fewer solver iterations and $1.8\times$ fewer bootstraps than the calibrated \texttt{ViT}, $99.4$ against $122.5$\,s per image, and $1.10\times$ faster than \texttt{ATLAS}. The gain is smaller than on \texttt{GPT-2} decode because the $80$ tokens (patches) of an image are processed at once, which moves a larger share of the per-image time into the linear families (\Cref{app:fig:hist-vit}); \texttt{HEAT} acts only on the iterative sites, so less of the wall time is within its reach. Classification is also a single forward pass. Even in this setting, \texttt{HEAT} yields the most stable circuit, with the fewest images failing during decryption.
\begin{figure}[h]
  \centering
  \includegraphics[width=0.5\textwidth]{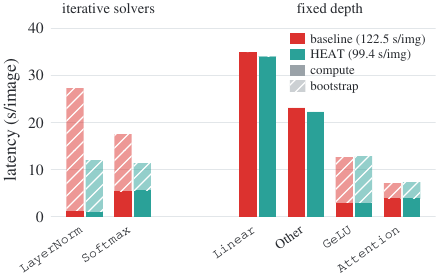}
  \caption{Per-family decode-latency composition, calibrated \texttt{ViT-B16} vs.\ \texttt{HEAT}: paired bars in absolute seconds per image, the hatched cap of each bar being the bootstrap time spent inside that family. \texttt{HEAT}'s saving concentrates the iterative methods' bars, where in this scenario the main source of latency is actual linears' computation and memory management.}
  \label{app:fig:hist-vit}
\end{figure}
\begin{table}
  \centering
  \small
  \setlength{\tabcolsep}{4pt}
  \caption{Encrypted \texttt{EuroSAT} classification at deployment (custom \texttt{A100-64GB}). $1024$ validation images drawn uniformly at random from the test set. it/fwd and bts/img: iterations of the iterative approximations and bootstraps per encrypted forward; s / image: wall time of the $12$ encoder blocks (Transf.) and end to end (e2e), the latter including the ${\approx}22$\,s head, median over the $1024$ images ($\pm$ one s.d.); fail.: images whose logits do not decrypt, counted as misses; top-1 under encryption over all $1024$ images, with the same circuit evaluated in plaintext on the full split as subscript.}
  \label{tab:vit-results}
  \begin{tabular}{@{}>{\columncolor{white}[0pt][\tabcolsep]}l cc ccc>{\columncolor{white}[\tabcolsep][0pt]}c@{}}
  \toprule
  & \multicolumn{2}{c}{circuit} & \multicolumn{2}{c}{s / image} & \multicolumn{2}{c}{$1024$ images} \\
  \cmidrule(lr){2-3}\cmidrule(lr){4-5}\cmidrule(lr){6-7}
  \textbf{System} & it/fwd\dn & bts/img\dn & Transf.\dn & e2e\dn & fail.\dn & top-1 $\texttt{FHE}_{(\textrm{plain})}$\up \\
  \midrule
  \rowcolor{gray!15}
  \,\texttt{ViT} (plaintext) & -- & -- & -- & -- & -- & 98.2 \\
  \cdashlinelr{1-7}
  \,\texttt{ViT} & 617 & 575 & 99.4$_{(\pm0.4)}$ & 122.5$_{(\pm0.6)}$ & 44/1024 & 94.4$_{(96.0)}$ \\
  \,$\mathcal{L}_\textrm{range}$ only & 529 & 464 & 88.3$_{(\pm0.3)}$ & 110.8$_{(\pm0.6)}$ & 21/1024 & 91.2$_{(93.0)}$ \\
  \,\texttt{ATLAS}~\citep{xie2026atlasautomatedapproximationtransformers} & 358 & 425 & 87.0$_{(\pm0.5)}$ & 109.7$_{(\pm1.4)}$ & 14/1024 & 96.7$_{(96.7)}$ \\
  \,\textbf{\texttt{HEAT} (ours)} & \textbf{250} & \textbf{326} & \textbf{77.2}$_{(\pm0.3)}$ & \textbf{99.4}$_{(\pm0.6)}$ & \textbf{3/1024} & \textbf{97.9}$_{(98.1)}$ \\
  \bottomrule
  \end{tabular}
\end{table}

\subsection{Downstream capability} \texttt{HEAT} changes the weights, and every perplexity in \Cref{tab:main-results,tab:ablation} is measured on OpenWebText, the fine-tuning corpus. \Cref{tab:glue-gpt2,tab:downstream-gpt2} therefore evaluate the deployed circuits in plaintext on standard zero-shot benchmarks. The calibrated circuit is within $0.1$ of the pretrained model on every metric. \texttt{HEAT} is within $1$ point on HellaSwag and PIQA and $1.8$ perplexity on C4.
\begin{table}[t]
  \centering
  \small
  \setlength{\tabcolsep}{7pt}
  \setlength{\aboverulesep}{0pt}\setlength{\belowrulesep}{0pt}
  \setlength{\extrarowheight}{2pt}
    \caption{Zero-shot GLUE accuracy for \texttt{GPT-2} small. The pretrained model, and the two deployed circuits evaluated in plaintext. \texttt{GPT-2} small sits at the majority-class floor on zero-shot GLUE, so the axis is whether \texttt{HEAT} loses anything. Its circuit matches or exceeds the calibrated one on $5/8$ tasks and on average.}
  \begin{tabular}{@{}l >{\columncolor{gray!15}}c cc@{}}
  \toprule
  & \texttt{GPT-2} & \multicolumn{2}{c}{\textbf{deployed circuit}} \\
  \cmidrule(lr){3-4}
  \textbf{Task}\,\up & (plaintext) & \texttt{GPT-2} & \textbf{\texttt{HEAT}} \\
  \midrule
  SST-2 & 55.0 & 55.0 & \textbf{56.1} \\
  MRPC & 56.1 & 56.1 & \textbf{67.4} \\
  QQP & 37.3 & \textbf{37.3} & 36.9 \\
  MNLI-m & 33.7 & 33.7 & \textbf{35.2} \\
  MNLI-mm & 33.2 & 33.2 & \textbf{34.8} \\
  QNLI & 50.2 & \textbf{50.2} & 49.5 \\
  RTE & 53.1 & \textbf{53.1} & 51.6 \\
  WNLI & 42.2 & 42.2 & \textbf{43.7} \\
  \cdashlinelr{1-4}
  \textbf{Average} & 45.1 & 45.1 & \textbf{46.9} \\
  \bottomrule
  \end{tabular}
  \label{tab:glue-gpt2}
\end{table}

\begin{table}[t]
    \centering
    \small
    \setlength{\tabcolsep}{6pt}
    \setlength{\aboverulesep}{0pt}\setlength{\belowrulesep}{0pt}
    \setlength{\extrarowheight}{2pt}
    \caption{Zero-shot downstream capability of \texttt{GPT-2} small. The pretrained model, and the deployed circuits evaluated in plaintext at their deployed iteration counts.}
    \begin{tabular}{@{}ll >{\columncolor{gray!15}}c cccc@{}}
    \toprule
    & & \texttt{GPT-2} & \multicolumn{4}{c}{\textbf{deployed circuit}} \\
    \cmidrule(lr){4-7}
    \textbf{Benchmark} & \textbf{Metric} & (plaintext) & \texttt{GPT-2} & \texttt{ATLAS} &
    \textbf{\texttt{HEAT}} & $\mathcal{L}_\textrm{range}$ only \\
    \midrule
    C4 & PPL\,\dn & 27.9 & \textbf{28.0} & 28.6 & 29.8 & 28.7 \\[2pt]
    \cdashlinelr{1-7}
    \multirow{2}{*}[-1pt]{HellaSwag} & acc\,\up & 28.9 & \textbf{28.9} & \textbf{28.9} & 28.8 & 28.7 \\
     & acc$_\textrm{norm}$\,\up & 31.1 & \textbf{31.1} & 31.0 & 30.5 & 31.0 \\[2pt]
    \cdashlinelr{1-7}
    \multirow{2}{*}[-1pt]{PIQA} & acc\,\up & 62.9 & \textbf{62.9} & 62.7 & 62.0 & 62.8 \\
     & acc$_\textrm{norm}$\,\up & 62.5 & 62.5 & \textbf{62.6} & 62.3 & \textbf{62.6} \\
    \bottomrule
    \end{tabular}

    \label{tab:downstream-gpt2}
\end{table}

\subsection{\texttt{ATLAS} replication}\label{app:atlas}
The \texttt{ATLAS} row of \Cref{tab:main-results} comes from its two-stage NSGA-II search, run with the released driver, operators and budget. We adapt the encoding to our circuit. The decision vector is the iteration counts, \emph{i.e.}\ per LayerNorm, Goldschmidt $\in[1,13]$ and Newton $\in[0,3]$; per Softmax, the Goldschmidt initialization $\in[2,9]$ and the refinement count $\in[1,12]$. Everything else stays at its calibrated value, \emph{e.g.}\ the fitted constants, the number of refinement passes, the exponential and \textrm{GELU} polynomials, and \texttt{CutMax}. Changing any of these requires re-fitting the constants, which is a new calibration rather than a search. \texttt{ATLAS}'s activation-degree axis has no counterpart because \textrm{GELU} places no bootstrap in any of our plans. The objectives are the total approximation depth, with levels per iteration measured on our circuit, and the mean absolute error of the last block's hidden states against the exact pretrained \texttt{GPT-2}, on $4\times512$ OpenWebText tokens in place of \texttt{ATLAS}'s $10$ sentences. We set the feasibility bound to $\varepsilon{=}1.0$, that admits every circuit within $2\%$ of the calibrated perplexity, and, as in \texttt{ATLAS}, the task re-evaluation makes the final choice. Stage one searches a single configuration shared by all layers as in the released code, with population $48$ for $50$ generations. Stage two, seeded with stage one's Pareto set, searches per layer with population $96$ for $225$ generations. The total is $24{,}144$ evaluations, that take ${\approx}46$ minutes on one \textrm{GPU}. We re-evaluate all $96$ points returned by stage two by plaintext perplexity on $8\times1024$ OpenWebText tokens, and keep the cheapest within $2\%$ of the calibrated circuit. Over three seeds, the cheapest such points sit at $438$, $468$ and $522$ iterations per forward; \Cref{tab:main-results} reports the cheapest. It cuts approximation depth from $838$ to $454$ ($-46\%$, against the ${\approx}35\%$ \texttt{ATLAS} reports on its own baseline) and iterations from $712$ to $438$.

\subsection{Where the encrypted drift comes from}\label{app:drift}
A bootstrap adds error to every slot, whatever the slot holds. We inject into plaintext models the error measured at the scheduled bootstraps of \Cref{tab:main-results}. The calibrated circuit loses $34\%$ perplexity with its KL rising $0.001\to0.45$, while the \texttt{HEAT} one loses nothing under the same noise (\Cref{fig:drift-noise}). Injecting the noise selectively to one approximation at the time identifies the Softmax as the main noise amplification and thus source of drift. The source of this is the refinement of \Cref{app:approx-operators}. Each square-and-normalize pass amplifies the relative error of its rows. The calibrated circuit's two to five passes per block push it well above $50\%$ by the end of the decode, whereas \texttt{HEAT}, whose halved score ranges and reduced iterations lets the calibrator use a single shallow pass, stays below $10\%$ (\Cref{fig:drift-chain}). The \texttt{ATLAS} circuit of \Cref{tab:main-results} inherits the calibrated refinement unchanged and, under the same injected noise, drifts like the calibrated circuit (\Cref{fig:drift-noise}). The injected noise gives between $2$ and $43$ collapsed windows of $64$ depending on the assumed noise floor, thus including the $13$ registered under actual \texttt{FHE} inference. The measured error is zero-mean with $\sigma=1.7\times10^{-4}$ per slot, independent of the plaintext amplitude up to ${\approx}30$; we resample it at every scheduled bootstrap, on $64$ held-out windows with three seeds, and fit no parameter to the encrypted runs. A search scored under such a noise model might close part of this gap; however it would be a new method rather than a baseline.
\begin{figure}[t]
    \centering
    \includegraphics[width=\figsingle]{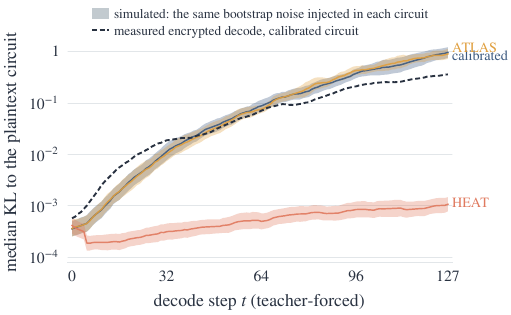}
    \caption{Median KL to the plaintext circuit when the same measured bootstrap noise is injected into the calibrated circuit, into \texttt{ATLAS} and into \texttt{HEAT}, over $64$ held-out OpenWebText windows of $128$ positions, disjoint from the data the \texttt{ATLAS} search used. Dashed: the measured encrypted decode of the calibrated circuit.}
    \label{fig:drift-noise}
\end{figure}
\begin{figure}[t]
    \centering
    \includegraphics[width=\figsingle]{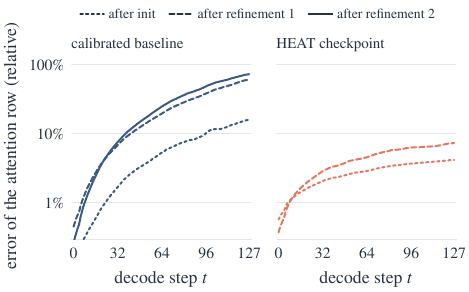}
    \caption{Relative error of the attention rows after each stage of the Softmax chain, noisy vs.\ noise-free. Left: calibrated circuit, two refinement passes. Right: \texttt{HEAT}, one.}
    \label{fig:drift-chain}
\end{figure}
\paragraph{Calibration tolerance.} Bootstrap can restore multiplicative depth while maintaining roughly $13$ bits of precision for the underlying payload, and thus a reasonable calibration choice would be $\varepsilon{=}10^{-4}$ as the target. This ensures the encrypted circuit operates at the maximum precision possible. \cref{tab:tolerance-ladder} shows what happens to the same pretrained models where the target is moved towards lower precision. None of the columns is alone enough to decide which is better as multiple factors must be taken into careful consideration. The plaintext perplexity of the circuit is roughly monotone with the precision, as one would expect; the encrypted circuit however has to be weighed between two factors: the collapsed rate and the top-1 accuracy. The first leaves only two candidates, namely $\{10^{-3}, 10^{-4}\}$, however the \texttt{FHE} perplexity of the $\varepsilon{=}10^{-3}$ option is unstable and far from its plaintext twin leaving only $10^{-4}$, \emph{i.e.}\ \texttt{FIDESlib} bts precision on the table.
\begin{table}
    \centering
    \small
    \setlength{\tabcolsep}{3pt}
    \caption{The calibrated circuit at four approximation tolerances, protocol and columns of Table~\ref{tab:main-results}. Loosening the tolerance buys iterations and time but not stability or fidelity. The ladder is not monotone (the $10^{-3}$ circuit collapses least and drifts most, the $10^{-1}$ circuit the reverse), every step keeps at least $2$ collapsed sequences and a $25$-$66\%$ perplexity loss against its own plaintext circuit, and none reaches \texttt{HEAT} on any column.}
    \begin{tabular}{@{}l cc cc c c c@{}}
    \toprule
    & \multicolumn{2}{c}{\textbf{circuit}\,\dn} & \multicolumn{2}{c}{\textbf{s/token}\,\dn} & \multicolumn{2}{c}{\textbf{fidelity}} & \textbf{perplexity}\,\dn \\
    \cmidrule(lr){2-3}\cmidrule(lr){4-5}\cmidrule(lr){6-7}\cmidrule(lr){8-8}
    \textbf{Calibration tolerance ($\varepsilon$)} & it/fwd & bts/tok & Transf. & \textbf{e2e} & coll.\,\dn & top-1\,\up & $\texttt{FHE}_{(\textrm{plain})}$ \\
    \midrule
    $10^{-1}$ & 571 & 418 & 36.6$_{(\pm0.3)}$ & 45.9$_{(\pm0.4)}$ & 16/64 & 60.9$_{(\pm3.4)}$ & 41.5$_{(33.1)}$ \\
    $10^{-2}$ & 642 & 456 & 40.1$_{(\pm0.3)}$ & 49.4$_{(\pm0.4)}$ & 14/64 & 60.1$_{(\pm2.9)}$ & 54.1$_{(32.6)}$ \\
    $10^{-3}$ & 678 & 474 & 41.5$_{(\pm0.3)}$ & 50.8$_{(\pm0.4)}$ & 2/64  & 70.5$_{(\pm1.7)}$ & 52.7$_{(31.7)}$ \\
    $10^{-4}$ (Table~\ref{tab:main-results}) & 712 & 513 & 45.0$_{(\pm0.3)}$ & 54.3$_{(\pm0.4)}$ & 7/64 & 70.3$_{(\pm2.6)}$ & 44.6$_{(31.7)}$ \\
    \cdashlinelr{1-8}
    \textbf{\texttt{HEAT}} & \textbf{228} & \textbf{326} & \textbf{28.2}$_{(\pm0.2)}$ & \textbf{37.6}$_{(\pm0.3)}$ & \textbf{0/64} & \textbf{82.9}$_{(\pm0.5)}$ & \textbf{30.7}$_{(31.0)}$ \\
    \bottomrule
    \end{tabular}
    \label{tab:tolerance-ladder}
\end{table}

\subsection{Scaling to larger models}
\texttt{HEAT} optimizes the circuit and it decides how many Goldschmidt, Newton and Softmax refinement steps every site needs, and that count is what we measure. Whether such a circuit is \emph{affordable} under encryption at a larger size is a different question, dominated by costs that \texttt{HEAT} cannot optimize, like key material, the weight plaintexts and their movement, and the packing tier, all of which grow with the width of the model (the 1280-wide residual stream of \texttt{GPT-2-large}, for instance, has to be padded to a 2048-slot layout). These are implementation problems on an axis orthogonal to the scope of this paper, which we don't explore. What can be verified without encryption is whether the iteration cut survives scale, because the plaintext simulation \emph{is} the deployed circuit and on \texttt{GPT-2} the two agree (\cref{sec:experiments}). \cref{tab:scaling} applies the recipe unchanged, apart from a smaller phase-2 learning rate on the largest model, to the \texttt{GPT-2} family up to 774M parameters, with perplexity measured on the same OpenWebText validation slice at full context (hence not comparable with the first-128-token numbers of the main matter). The cut is nearly flat with $3.1\times$, $3.0\times$ and $2.8\times$ fewer iterations per forward pass than the calibrated circuit, at a perplexity between $-2.2\%$ and $+4.6\%$ of the calibrated baseline and at half the cost of the range-regularized (squeezed) circuit. Larger models keep the lower perplexity floor their size affords, so the saving is not bought with the model's capacity.

\begin{table}[t]
  \centering
  \small
  \setlength{\tabcolsep}{5pt}
  \setlength{\aboverulesep}{0pt}\setlength{\belowrulesep}{0pt}
  \setlength{\extrarowheight}{2pt}
    \caption{Applying the \texttt{HEAT} fine-tuning paradigm to the whole \texttt{GPT-2} family up to 774M parameters confirming an almost flat $3\times$ iterations gain while retaining performance. (Bigger models aren't tested in the \texttt{FHE} inference but just the plaintext scenario).}
  \begin{tabular}{@{}ll >{\columncolor{gray!15}}r >{\columncolor{gray!15}}r rr rr r@{}}
  \toprule
  & & \multicolumn{2}{c}{\cellcolor{white}\texttt{calibrated}} & \multicolumn{2}{c}{squeezed} & \multicolumn{2}{c}{\textbf{\texttt{HEAT}}} & \\
  \cmidrule(lr){3-4}\cmidrule(lr){5-6}\cmidrule(lr){7-8}
  \textbf{Model} & \textbf{params} & it/fwd\,\dn & ppl\,\dn & it/fwd\,\dn & ppl\,\dn & it/fwd\,\dn & ppl\,\dn & \textbf{reduction} \\
  \midrule
  \texttt{GPT-2}        & 124\,M & 712  & 23.72 & 455  & 24.19 & \textbf{228} & \textbf{23.19} & \textbf{3.12$\times$} \\
  \texttt{GPT-2-medium} & 355\,M & 1319 & 16.76 & 800  & 17.05 & \textbf{441} & \textbf{16.94} & \textbf{2.99$\times$} \\
  \texttt{GPT-2-large}  & 774\,M & 1897 & 15.17 & 1278 & 15.86 & \textbf{672} & \textbf{15.87} & \textbf{2.82$\times$} \\
  \bottomrule
  \end{tabular}

  \label{tab:scaling}
\end{table}

\subsection{Free-running encrypted generation}\label{app:free-running}
Every fidelity number in the discussion is teacher-forced, \emph{i.e.}\ each position receives the ground-truth token, so an encrypted error cannot propagate to the next step. We therefore also run the deployed \texttt{HEAT} circuit the way a server would, autoregressively and fully under encryption. The prompt is encrypted and processed token by token through the decode path of \Cref{tab:main-results}. From then on, each token is selected by the encrypted $\arg\max$ of \Cref{app:approx-operators}, whose one-hot output is mapped homomorphically to its token and positional embedding and fed back, after a bootstrap, as the next input. \Cref{app:fig:free-running} compares $32$ encrypted tokens per prompt with the greedy decoding of the same circuit (the \texttt{HEAT} weights at the learned iteration counts) evaluated in plaintext. On three of the four prompts the two are identical token for token. On the fourth they branch at the third token, where the plaintext circuit ranks ``the'' over ``a'' by $0.005$ in logit and encryption swaps them, and both continuations remain fluent. Scored step by step against the plaintext circuit run on the encrypted prefix, the encrypted $\arg\max$ agrees at $126$ of $128$ steps, both misses being near-ties of this kind, with a median \textrm{KL} of $5\cdot10^{-4}$ to $8\cdot10^{-3}$ per prompt. For reference, the greedy continuation of the pretrained \texttt{GPT-2} departs from \texttt{HEAT}'s within the first seven tokens of every prompt, at the same token whether \texttt{HEAT} runs encrypted or in plaintext meaning that the text changes with the fine-tuning, not with the encryption. Greedy decoding of a model of this size falls into repetition within a sentence or two, and the classic remedies, repetition penalties and sampling, are not yet available for encrypted free-running generation, whose every token is the encrypted $\arg\max$. We therefore choose not to display the repetitions and cut each continuation where it starts repeating, while every statistic above is computed over all $32$ tokens. The length of decoding is fixed in advance as, by default, encrypted generation runs for a pre-set number of steps, since the server never sees the generated tokens and cannot stop on an end-of-sequence token.
\begin{figure}[h]
  \centering
  \small
  \newcommand{\nl}{{\color{gray}$\hookleftarrow$}\,}
  \newcommand{\meta}[1]{{\footnotesize\color{gray!80!black}#1}}
  \newcommand{\cut}{{\color{gray!80!black}\,[\ldots]}}
  \newcommand{\genbox}[5]{\fcolorbox{gray!70}{gray!5}{\parbox[t][#1][s]{\dimexpr0.333\linewidth-2\fboxsep-2\fboxrule-0.5em\relax}{\raggedright\textbf{#2}\par\smallskip\meta{\texttt{HEAT}, encrypted $=$ plaintext}\\\textcolor{cplx}{#3}\par\medskip\meta{pretrained \texttt{GPT-2}}\\\textcolor{gray}{#4}\par\vfill\meta{#5}}}}
  \genbox{6.4cm}{The capital of France is}{the capital of the French Republic. The French Republic is the largest state in the world.\cut}{the capital of the French Republic, and the capital of the French Republic is the capital of the French Republic.\cut}{agree $32/32$ \,$\cdot$\, \textrm{KL} $0.0023$}\hfill
  \genbox{6.4cm}{Recently, scientists have discovered that}{the brain is a complex system of neurons, which are connected to the brain through the spinal cord.\nl\nl The researchers found that\cut}{the brain's ability to process information is impaired when it is exposed to a high-energy source of energy.\nl\nl The findings, published in the journal Nature}{agree $32/32$ \,$\cdot$\, \textrm{KL} $0.0014$}\hfill
  \genbox{6.4cm}{The best way to learn a new language is}{to learn it from a teacher.\cut}{to learn it from a native speaker.\cut}{agree $32/32$ \,$\cdot$\, \textrm{KL} $0.0005$}\\[0.6em]
  \fcolorbox{gray!70}{gray!5}{\parbox{\dimexpr\linewidth-2\fboxsep-2\fboxrule\relax}{\raggedright\textbf{It should also be noted that early on in this case, Hotfile}\par\smallskip
    \begin{minipage}[t]{0.32\linewidth}\raggedright\meta{\texttt{HEAT}, encrypted}\\\textcolor{cplx}{was not \underline{a} legitimate source for the files.\cut}\end{minipage}\hfill
    \begin{minipage}[t]{0.32\linewidth}\raggedright\meta{\texttt{HEAT}, plaintext}\\\textcolor{gray}{was not \underline{the} only site that was targeted by the FBI. The FBI also targeted the site of the same name,\cut}\end{minipage}\hfill
    \begin{minipage}[t]{0.32\linewidth}\raggedright\meta{pretrained \texttt{GPT-2}}\\\textcolor{gray}{was not a file server. It was a file server that was used to store files on the server.\cut}\end{minipage}\par\medskip
    \meta{branches from its plaintext at a near-tie \,$\cdot$\, agree $30/32$ \,$\cdot$\, \textrm{KL} $0.0076$}}}
  \caption{Free-running generation of the deployed \texttt{HEAT} circuit under \texttt{FHE}, $32$ greedy tokens per prompt. The prompt (black) is encrypted, and every token \texttt{HEAT} generates (blue) is chosen by the encrypted $\arg\max$ and fed back without decryption. In gray it's the plaintext greedy continuations, of the pretrained \texttt{GPT-2} for reference.}
  \label{app:fig:free-running}
\end{figure}

\section{Societal Impacts}\label{app:societal}
\texttt{FHE} brings an additional dimension to the societal implications of deploying modern deep learning systems. While privacy-preserving inference can provide substantial benefits by reducing the exposure of sensitive user data, the privacy guarantees offered by \texttt{FHE} may also change how such systems can be deployed and used in practice. Deep learning models already have broad societal benefits, while at the same time lowering the barriers to a range of potentially harmful applications. Consequently, the implications of combining increasingly capable models with strong cryptographic privacy guarantees warrant careful consideration.

In particular, there remains a gap between the privacy properties demonstrated by \texttt{FHE} protocols in controlled settings and those of complete systems deployed in practice. Practical deployments involve additional components, assumptions, and operational choices that may affect the resulting privacy and security properties. Moreover, naively incorporating \texttt{FHE} into an existing inference pipeline does not necessarily provide the same guarantees as a carefully designed end-to-end private system. We therefore believe that the broader societal implications of deploying \texttt{FHE}-based inference systems, including both their potential benefits and unintended consequences, deserve further investigation.

\end{document}